\documentclass[11pt, letterpaper]{amsart}

\usepackage{color}
\usepackage[english]{babel}
\usepackage[latin1]{inputenc}
\usepackage{times}
\usepackage[T1]{fontenc}
\usepackage{amsfonts}
\usepackage{epsfig,psfrag,latexsym}
\usepackage{amsmath, amssymb}
\usepackage{graphicx}
\usepackage{paralist}
\usepackage{enumerate}
\usepackage{amsthm}

\newtheorem{theorem}{Theorem}[section]

\newtheorem{lemma}[theorem]{Lemma}
\newtheorem{proposition}[theorem]{Proposition}
\theoremstyle{definition}
\newtheorem{definition}[theorem]{Definition}

\newtheorem{assumption}[theorem]{Assumption}

\theoremstyle{remark}
\newtheorem{remark}[theorem]{Remark}

\numberwithin{equation}{section}

\newcommand{\reals}{\mathbb R}

\newcommand{\eps}{\varepsilon}

\newcommand{\such}{\ | \ }
\newcommand{\xpn}[1]{\exp\left(#1\right)}

\newcommand{\prob}{\mathbb{P}}
\newcommand{\qprob}{\mathbb{Q}}

\newcommand{\qprobalt}[2]{\qprob^{#1}\bra{#2}}
\newcommand{\condprobalt}[3]{\prob^{#1}\bra{#2\bigg|#3}}
\newcommand{\qcondprobalt}[3]{\qprob^{#1}\bra{#2\bigg|#3}}
\newcommand{\esp}{\mathbb{E}}
\newcommand{\espalt}[2]{\esp^{#1}\bra{#2}}
\newcommand{\condespalt}[3]{\esp^{#1}\bra{#2\big|#3}}

\newcommand{\F}{\mathcal{F}}

\newcommand{\G}{\mathcal{G}}

\newcommand{\filt}{\mathbb{F}}

\newcommand{\nada}[1]{}

\newcommand{\dfn}{\, := \,}

\newcommand{\bra}[1]{\left[#1\right]}
\newcommand{\cbra}[1]{\left\{#1\right\}}
\newcommand{\dbra}[1]{[\kern-0.15em[ #1 ]\kern-0.15em]}
\newcommand{\dbraco}[1]{[\kern-0.15em[ #1 [\kern-0.15em[}

\newcommand{\ol}[1]{\overline{#1}}

\newcommand{\Kn}{\mathbb{K}_n}
\newcommand{\sn}{\mathcal{A}^n}
\newcommand{\nst}{\tilde{n}}
\newcommand{\alphanorm}[2]{\| #1\|_{\alpha,\ol{D}_{#2}}}
\newcommand{\calphanorm}{C(n,\alphanorm{\eta_1}{n},\alphanorm{\eta_2}{n})}
\newcommand{\wt}[1]{\widetilde{#1}}
\newcommand{\filtwt}[1]{\filt^{\wt{#1}}}

\title{Endogenous Current Coupons}

\author{Zhe Cheng}
\address{Carnegie Mellon University, Morgan Stanley}
\email{joshua.z.cheng@gmail.com}

\author{Scott Robertson}
\address{Department of Mathematical Sciences\\
Carnegie Mellon University\\
Pittsburgh, PA 15213}
\email{scottrob@andrew.cmu.edu}
\thanks{S. Robertson is supported in part by the National Science Foundation
  under grant number DMS-1312419.}

\date{\today}

\begin{document}

\begin{abstract}
We consider the problem of identifying current coupons for Agency backed To-be-Announced (TBA) Mortgage Backed Securities.  In a doubly stochastic factor based model which allows for prepayment intensities to depend upon current and origination mortgage rates, as well as underlying investment factors, we identify the current coupon with solutions to a degenerate elliptic, non-linear fixed point problem.  Using Schaefer's theorem we prove existence of current coupons.  We also provide an explicit approximation to the fixed point, valid for compact perturbations off a baseline factor-based intensity model. Numerical examples are provided which show the approximation performs remarkably well in estimating the current coupon.
\end{abstract}

\maketitle


\section{Introduction}\label{S:intro}

The goal of this paper is to prove existence of \emph{endogenous} mortgage origination rates, defined as those which yield par-valued mortgage pools. For Agency backed (e.g. FNMA, FHLMC, GNMA) To-be-Announced (TBA) pools of residential mortgages, such rates are also called \emph{current coupons}. In addition to proving existence of current coupons, we wish to provide a fast, easy to implement, and accurate way of computing the current coupon, as it is well known (see \cite{MR2352905,goncharov2009computing}) that iterative, monte-carlo or partial differential equation based, methods are prohibitively time-consuming to implement.

The residential mortgage market is currently the largest segment of the US fixed income market (see \cite{Fed_Flow_Q2_2015}) and the problem of pricing Mortgage Backed Securities (MBS) is of significant financial interest. The primary difficulty in pricing MBS, however, is the fact that the home buyer has, at any time prior to maturity of the loan, the right to prepay all or part of her mortgage with few, if any, penalties.  In particular, the mortgagee may refinance (multiple times) her loan in order to take advantage of current market conditions.  Adding to the complication is the well known fact that individual mortgagors vary in their financial sophistication and often do not prepay optimally. For example many mortgagors delay their refinancing decisions even when interest rates decline to a level such that it is financially optimal to refinance (see \cite{stanton1995rational}).  

Agency backed MBS has been the major component of the MBS market since the financial crisis. Issuance of agency MBS has remained robust since 2007 while mortgage securitization by private financial institutions has declined to very low levels (see \cite{Mortgage_Data_SIFMA}). A well-known feature of agency MBS is that each bond carries either an explicit government credit guarantee, or is perceived to carry an implicit one. Agency MBS investors are thus protected from credit losses in case of mortgage borrower default, and as such, for valuation purposes, defaults appear to the pool holder nearly identical to prepayments.

Another less well-recognized feature of agency MBS is that more than 90 percent of agency MBS trading volume occurs in a liquid forward market, known as the TBA market (see \cite{TBA_Fact_Sheet}). The distinguishing feature of a TBA trade is that the actual identity of the securities to be delivered on the settlement date is not specified on the trade date. Instead, the buyer and the seller agree upon general parameters of the securities to be delivered, such as issuer, maturity, coupon, price, par amount and settlement date. Closely related to TBA mortgage-backed securities is the secondary-market MBS rate, known as the current coupon. The current coupon is a coupon rate interpolated from the observed TBA prices that makes the price of a TBA with current delivery month equal to par. As such, the current coupon is an endogenous rate, and current coupon rates are widely used as a benchmark for MBS pool valuation, playing a key role in the secondary mortgage market.

Broadly speaking, within the academic literature, there are two methods used to valuate MBS: the ``option theoretic'' and ``reduced form'' methods (see \cite{MR2352905,MR2260051} for a more thorough introduction and literature review). The option theoretic method treats the right to prepay as an American style embedded option and MBS valuation is performed using options pricing theory.  Early results along this line were obtained in \cite{1981_Dunn_Mcconnell, kau1995valuation,kalotay2004option}.  However, it was quickly recognized that option theoretic methods suffer due to the non-optimal prepayment behavior of borrowers, and hence the option theoretic approach has not been widely adopted by mortgage market practitioners.

Alternatively, the reduced form method borrows from the theory of credit derivative valuation and assumes prepayments are driven by an underlying intensity process which may be estimated from historical data. Here, the non-optimality of prepayment behavior is built into the intensity function. Reduced form methods have been studied in \cite{schwartz1989prepayment, pliska2006mortgage,kau1995valuation,deng2000mortgage, MR2352905,MR2943181,MR2260051,goncharov2009computing,MR2667899} amongst others. In this paper, we consider the reduced form method. We pay particular attention to \cite{MR2352905}, which computes rates when the intensity is driven by one (or many) economic factors and \cite{MR2943181}, which considers similar intensities to those we treat. Further connections with \cite{MR2943181} are discussed below.

Aside from the amortizing nature of a mortgage loan, the key difference between MBS and credit derivative valuation is the dependence of the mortgage pool value on the mortgage origination rate. Indeed, one has the heuristic relationship
\begin{equation*}
\textrm{Mortgage Rate: } m_0 \quad\Longrightarrow \ \textrm{Prepayment Time: }\tau(m_0)\quad \Longrightarrow \ \textrm{Pool Value: } M(m_0).
\end{equation*}
Thus, there is a natural and delicate fixed point problem in finding $m_0$ so that $M(m_0)$ is par valued. In reduced form models, this circular dependence is captured in the intensity function. This is in contrast to credit valuation, where one typically expresses the default intensity $\gamma$ as a function of the underlying economic factors, or state variables $X$. Indeed, whereas an intensity specification $\gamma_t = \gamma(X_t)$ may be appropriate for credit derivatives, for MBS valuation, it is desirable to allow $\gamma$ to additionally depend upon both the mortgage origination rate $m_0$ and the current mortgage rate $m_t$ available for refinancing: i.e. $\gamma_t = \gamma(X_t,m_0,m_t)$. Thus, in a time-homogeneous Markovian setting one hypothesizes that $m_t = m(X_t)$ is a function of the underlying economic factors and hence
\begin{equation}
\gamma_t = \gamma(X_t,m(X_0),m(X_t)).
\end{equation}
With this specification, the goal is then to find a \emph{current coupon function} $m$ so that the pool value $M(m(X_0)) = 1$ for all values $X_0$.

\cite{pliska2006mortgage} and \cite{MR2260051,MR2943181} first incorporated the endogenous mortgage rate into an intensity-based framework, taking into account the dependence of $\gamma$ on $m$. In particular, \cite{MR2943181} presented a proof of the existence of a current coupon in a diffusion model similar to that presently considered. However, we wish to point out three key differences between \cite{MR2943181} and the present work.  First and foremost, there is an error in \cite{MR2943181} (Proposition 4.1 therein is evidently incorrect for the discontinuous intensities considered) which, while not necessarily invalidating the main results, certainly calls them into question. Second, the existence proof, based on a so-called "Lebesgue set method", is highly non-standard, whereas our proof of existence uses standard topological fixed point theorems.  Third, our method of proof has the added benefit that we are able to show regularity in the current coupon function, whereas in \cite{MR2943181} only measurable solutions are obtained.

Equally important as identifying existence of current coupons is actually computing the current coupon. Indeed, a naive application of the contraction principle where one fixes an initial function $m_0$ and then sets $m_n(X_0) = M(m_{n-1}(X_0)), n=1,2,...$ with the idea that $m_n\rightarrow 1$, while not only theoretically unjustified, is also prohibitively slow.  To overcome this problem, \cite{MR2352905} writes the intensity as solely a function of the underlying factors with the idea that this captures the bulk of prepayments. Then, for CIR interest rates, the endogenous rate is rapidly computed using eigen-function expansions. In \cite{goncharov2009computing} a non-iterative method is proposed borrowing ideas from partial differential equations theory.  In the current paper we take an alternate approach, approximating the current coupon via perturbation analysis. Thus uses the well known fact (see \cite{MR2943181}) that unique current coupon functions exist when $\gamma_t= \gamma(X_t)$ only depends upon the factors. Specifically, we note that one may always write
\begin{equation*}
\gamma(x,m,z) = \gamma_0(x) + \gamma_1(x,m,z),
\end{equation*}
by taking $\gamma_0(x) \equiv 0$, but also in the case where the full intensity is assumed to be a constant intensity $\gamma>0$ plus an additional component.  We then embed this decomposition via
\begin{equation*}
\gamma^\eps(x,m,z) = \gamma_0(x) + \eps \gamma_1(x,m,z);\qquad \eps > 0.
\end{equation*}
For $\eps =0$, there is a unique current coupon function $m_0(x)$. Sending $\eps\rightarrow 0$ we obtain a unique, explicit, closed form expression for $m_1(x)$ so that $m^{\eps}(x) = m_0(x)+\eps m_1(x) + o(\eps)$.  With this decomposition, valid for any continuous fixed point $m^{\eps}$ we naturally consider the numerical approximation (at $\eps =1$) of $m(x)\approx m_0(x) + m_1(x)$.  It turns out this approximation does very well in practice: differing by $\leq 10$ basis points (on absolute rate levels of $4\%-12\%$) from the theoretical fixed point determined by naive contraction.

The rest of the paper is organized as follows. In Section \ref{S:ccpn} we give a heuristic derivation of the fixed point problem.  Section \ref{S:model} specifies the fixed point problem to a Markovian framework where $X$ is a non-explosive locally elliptic diffusion on a general state space in $\reals^d$, making precise assumptions on the model coefficients, as well as the intensity function. Section \ref{S:model} culminates with Theorem \ref{T:main_result} which proves existence of a current coupon function, under the assumption that $\gamma(x,m,z)$ is approximately constant in $m$ for large values of $m$ (see Remark \ref{R:gamma} for more discussion on our main assumption).  Section \ref{S:perturb} performs the perturbation analysis with Theorem \ref{T:perturb} explicitly identifying the leading order terms in the expansion. Section \ref{S:numerical} gives a numerical example where the current coupon approximated via perturbation analysis is compared to the function obtained through naive contraction. Appendices \ref{S:proof_of_main} -- \ref{S:qprob} contain the proofs. In particular, as the mortgage market is typically incomplete, a rigorous construction of the particular risk neutral measures used here for pricing is given.  Aside being done for the sake of mathematical rigor, we show that when pricing the mortgage pool, one may assume the intensity processes coincide between the physical and risk neutral measures and hence can be estimated using observed prepayment data.

\section{Endogenous Current Coupons}\label{S:ccpn}

Consider a level-payment, fully amortized $T$-year fixed rate mortgage which is originated at time $t=0$. The mortgagor thus takes a loan of $P_0$ dollars at origination and pays a continuous coupon stream at the constant rate of $c>0$ dollars per annum during the lifetime of the mortgage $[0,T]$. The interest is compounded at the constant mortgage rate $m$ fixed at origination. In the absence of prepayments, the scheduled outstanding principal of the mortgage, denoted by $p(t,m)$ for $0 \leq t \leq T$ and $m \geq 0$, satisfies the following ordinary differential equation (ODE):
\begin{equation}\label{E:balance_no_prepay}
p_t(t,m) = mp(t,m)-c;\qquad p(0,m)=P_0, \ p(T,m)=0,
\end{equation}
where $p_t$ is the partial derivative with respect to $t$. \eqref{E:balance_no_prepay} has solution
\begin{equation}\label{E:balance_closed_form}
p(t,m)= P_{0}\dfrac{1-e^{-m(T-t)}}{1-e^{-mT}};\ (m>0), \qquad p(t,m) = P_{0}\left(1-\frac{t}{T}\right); \ (m=0).
\end{equation}
Since $P_0$ factors out of the above equation, we assume $P_0 = 1$ throughout so that
\begin{equation}\label{E:balance_closed_form_P01}
p(t,m)=\frac{1-e^{-m(T-t)}}{1-e^{-mT}};\ (m>0),\qquad p(t,m) = \left(1-\frac{t}{T}\right);\ (m=0).
\end{equation}
From \eqref{E:balance_no_prepay} and \eqref{E:balance_closed_form_P01} we can express the coupon stream payment $c$ in terms of $m$ and $T$ as well:
\begin{equation}\label{E:coupon}
c = c(m) = \frac{m}{1-e^{-mT}}; \ (m>0),\qquad c(m) = \frac{1}{T}; \ (m=0).
\end{equation}

We first informally derive a fixed point equation for the current coupon $m$.  This argument will be made rigorous in Section \ref{S:model} and Appendix \ref{S:qprob} below.  In the absence of prepayments, the mortgage balance $p(t,m)$ evolves according to \eqref{E:balance_closed_form_P01}. Consider now when there is a (random) prepayment time $\tau$ under a pricing measure $\qprob$ (here, the underlying probability space is $(\Omega,\G,\qprob)$). In other words, if $\tau\leq T$, the owner of the mortgage at time $\tau$ prepays the remaining balance $p(\tau,m)$.  Assuming an interest rate $r = \cbra{r_t}_{t\leq T}$ the value of the mortgage is
\begin{equation}\label{E:mortgage_value}
M(m) = \espalt{\qprob}{\underbrace{\int_0^{\tau\wedge T} c(m)e^{-\int_0^t r_u du}dt}_{\textrm{Coupon Payments}} + \underbrace{1_{\tau\leq T} p(\tau,m)e^{-\int_0^\tau r_u du}}_{\textrm{Prepayment}}}.
\end{equation}
Next, assume that the interest rate process is adapted to a filtration $\filt = \cbra{\F_t}_{t\leq T}$ where $\F = \vee_{t\leq T} \F_t \subset \G$ and that $\tau$ has an intensity $\gamma = \cbra{\gamma_t}_{t\leq T}$ with respect to $(\qprob,\filt)$:
\begin{equation}\label{E:intensity_equation}
\qcondprobalt{}{\tau>t}{\F}\footnote{This equality requires an additional hypotheses on how $\tau$ is constructed and will be shown to hold in the current setup.} = \qcondprobalt{}{\tau>t}{\F_t} = e^{-\int_0^t\gamma_u du}\qquad t\geq 0,
\end{equation}
for some non-negative, integrable, adapted process $\gamma$. From this, we obtain (see \cite{MR2943181,MR2352905}) the value of the mortgage as
\begin{equation}\label{E:mortgage_value_nice}
M(m) = 1 + \espalt{\qprob}{\int_0^T p(t,m)(m-r_t)e^{-\int_0^t(r_u+\gamma_u)du}dt}.
\end{equation}

\nada{
\begin{lemma}\label{L:mortgage_nice}
The value of the mortgage is
\begin{equation}\label{E:mortgage_value_nice}
M(m) = 1 + \espalt{\qprob}{\int_0^T p(t,m)(m-r_t)e^{-\int_0^t(r_u+\gamma_u)du}dt}.
\end{equation}
\end{lemma}

\begin{proof}[Proof of Lemma \ref{L:mortgage_nice}]
Regarding the coupon payments in \eqref{E:mortgage_value} we have
\begin{equation*}
\begin{split}
\espalt{\qprob}{\int_0^{\tau\wedge T}c(m)e^{-\int_0^t r_u du}dt} &= \espalt{\qprob}{\int_0^T(mp(t,m)-p_t(t,m))e^{-\int_0^t(r_u+\gamma_u)du}dt},\\
&= 1 + \espalt{\qprob}{\int_0^T p(t,m)(m-r_t-\gamma_t)e^{-\int_0^t(r_u+\gamma_u)du}dt},
\end{split}
\end{equation*}
where the second equality uses \eqref{E:balance_no_prepay} and the third equality follows from integrating by parts and using \eqref{E:balance_no_prepay} again (with $P_0=1$). For the prepayment terms in \eqref{E:mortgage_value} we obtain, using \eqref{E:intensity_equation}:
\begin{equation*}
\begin{split}
\espalt{\qprob}{1_{\tau\leq T}p(\tau,m)e^{-\int_0^\tau r_udu}} &= \espalt{\qprob}{\int_0^\infty 1_{t\leq T}p(t,m)e^{-\int_0^t r_udu}\qcondprobalt{}{\tau\in dt}{\F}},\\
&= \espalt{\qprob}{\int_0^T p(t,m)\gamma_t e^{-\int_0^t(r_u+\gamma_u)du}dt}.
\end{split}
\end{equation*}
Putting the above results in \eqref{E:mortgage_value} gives the result.
\end{proof}
}
The mortgage rate $m$ is said to be \emph{endogenous} if $M(m) = P_0 = 1$.  In view of \ref{E:mortgage_value_nice}, we seek $m$ so that
\begin{equation}\label{E:m_goal}
0 = \espalt{\qprob}{\int_0^T p(t,m)(m-r_t)e^{-\int_0^t(r_u+\gamma_u)du}dt}.
\end{equation}

\section{The Model and Fixed Point Problem}\label{S:model}

The above analysis is now specified to a doubly stochastic, intensity based model for the mortgage prepayment time $\tau$.  To make this precise, fix a probability space $(\Omega,\G,\qprob)$. We first remark:

\begin{remark}\label{R:rnm}
The measure $\qprob$ is interpreted as a pricing, or risk neutral, measure and we write $\espalt{}{\cdot}$ for $\espalt{\qprob}{\cdot}$ throughout. In Appendix \ref{S:qprob} we offer two rigorous constructions of $\qprob$: one valid for a ``large'' pool and one valid for a single loan pool.  In particular we will show that when estimating the prepayment intensity function $\gamma$ described in Assumption \ref{A:gamma} below, one may use observed prepayment data rather than estimating prepayments under the particular risk neutral measure $\qprob$. For ease of exposition, however, we delay this construction,  simply assuming  a mortgage rate $m$ is the current coupon if it satisfies \eqref{E:m_goal}.
\end{remark}

Let $W$ be a standard, d-dimensional Brownian motion under $\qprob$. The underlying economic factors which affect prepayments are governed by the process $X$ satisfying the stochastic differential equation (SDE)
\begin{equation}\label{E:factors}
dX_t = b(X_t)dt + a(X_t)dW_t.
\end{equation}
The state space of $X$ is an open, connected region $D\subseteq \reals^d$ which satisfies
\begin{assumption}\label{A:region}
$D = \cup_{n=1}^\infty D_n$ where for each $n$, $D_n$ is open and bounded with smooth boundary.  Furthermore, $\bar{D}_n\subset D_{n+1}$.
\end{assumption}
Regarding the coefficients in \eqref{E:factors} we assume that $b:D\mapsto \reals^d$ and let $A:D\mapsto\mathbb{S}^d_{++}$, the space of symmetric positive definite $d\times d$ matrices.  We then take $a=\sqrt{A}$, the unique positive definite symmetric square root of $A$. We assume $b,A$ satisfy the following regularity and local-ellipticity assumptions
\begin{assumption}\label{A:factor_coefficients}\text{}
\begin{enumerate}[1)]
\item $A$ is locally elliptic: i.e. for each $n$ there exists $K_1(n)>0$ so that for all $\xi\in\reals^d\setminus \cbra{0}$ and $x\in D_n$ we have $\xi'A(x)\xi \geq K_1(n)\xi'\xi$.
\item $b$ and $A$ are locally Lipschitz with Lipschitz constant $K_2(n)$.
\end{enumerate}
\end{assumption}
Assumption \ref{A:factor_coefficients} implies existence of a local solution solution to the SDE in \eqref{E:factors}.  To ensure existence of a global solution we assume the process does not explode: i.e.
\begin{assumption}\label{A:no_explosion}
For all $x\in D$ and $T>0$, we have $\qprobalt{x}{X_t\in D, \ \forall\ t\leq T} = 1$, where $\qprob^x$ denotes the conditional probability given $X_0 = x$.
\end{assumption}
Under Assumptions \ref{A:factor_coefficients}, \ref{A:no_explosion} it follows that $X$ has a unique strong solution.  Furthermore, since the short term interest rate $r$ plays a key role in the mortgage evaluation, we assume the first coordinate of $X$ is the interest rate: i.e. $X^{(1)}_t = r_t$ and that the state space of $X^{(1)}$ is $(0,\infty)$: i.e.

\begin{assumption}\label{A:pos_ir}
The state space of $r \dfn X^{(1)}$ is $(0,\infty)$.
\end{assumption}

To precisely define the intensity $\gamma$ in \eqref{E:m_goal} we adopt the following methodology. Let $m:D\mapsto [0,\infty)$ be a given candidate current coupon function, in that we wish for $m(x)$ to be the endogenous current coupon given $X_0 =x\in D$.  As mentioned in the introduction, we hypothesize $\gamma$ is a function of
\begin{enumerate}[$\bullet$]
\item The underlying factor process $X$.
\item The contract mortgage rate $m(x)$.
\item The current mortgage rate available via refinancing $m(X)$\footnote{Technically we should allow $m$ to be time-dependent as well: i.e. $m_t = m(t,X_t)$ but, due to the time-homogeneity of the diffusion $X$, it suffices to consider $m_t=m(X_t)$.}.
\end{enumerate}
Thus, at time $t\leq T$ we have $\gamma_t = \gamma(X_t,m(x), m(X_t))$, where $\gamma:D\times [0,\infty)\times [0,\infty)$ is an exogenously defined function. To facilitate our main assumption on $\gamma$ we first define the auxiliary function
\begin{equation}\label{E:Xi_def}
\Xi(x)\dfn \inf_{0<\beta<1}\frac{\beta e^{-\beta x}}{(1-\beta)(1-e^{-\beta x})};\qquad x >0.
\end{equation}
Straightforward analysis shows that $\Xi$ is decreasing with $x$ and
\begin{equation}\label{E:Xi_lims}
\Xi(x) =\frac{1}{x} \textrm{ for } x\leq 2;\qquad \lim_{x\uparrow\infty}\frac{\Xi(x)}{xe^{-(x-1)}}=1.
\end{equation}
With this definition, we make the following assumptions regarding $\gamma$. To ease presentation, define $E\dfn D\times(0,\infty)\times(0,\infty)$ and $E_n\dfn D_n\times (0,n)\times (0,n), n\in\mathbb{N}$.

\begin{assumption}\label{A:gamma} Assume $\gamma: E\mapsto [0,\infty)$ satisfies
\begin{enumerate}[1)]
\item $\gamma\in C^{2}(E)$ and for each $n$, the derivatives of order $\leq 2$ can be continuously extended to $\bar{E}_n$ \footnote{Henceforth we will assume $\gamma$ and its derivatives of order $\leq 2$ are defined on $D\times [0,\infty)\times[0,\infty)$ with the values at zero being the continuous extensions.}, and are Lipschitz continuous on $\bar{E}_n$ with Lipschitz constant $L_\gamma(n)$.

\item $\gamma(x,m,z)$ and $\gamma_m(x,0,z)$ are locally bounded in $x$, uniformly in $(m,z)$ and $z$ respectively. I.e. for each $n$ there is a $B_\gamma(n)>0$ so that
\begin{equation}\label{E:gamma_loc_x_bdd}
\sup_{x\in D_n,m,z\geq 0}\gamma(x,m,z)\leq B_\gamma(n);\qquad \sup_{x\in D_n,z\geq 0} \gamma_m(x,0,z) \leq B_\gamma(n).
\end{equation}

\item With $\Xi$ as in \eqref{E:Xi_def}, it holds that
\begin{equation}\label{E:gamma_m_bdd}
\begin{split}
0 \leq \gamma_m(x,m,z)\leq \Xi(mT);\quad x\in D, m,z\geq 0.
\end{split}
\end{equation}
\end{enumerate}
\end{assumption}

\begin{remark}\label{R:gamma}
Regarding Assumption \ref{A:gamma}, that $\gamma\geq 0$ is standard. The local regularity conditions are not overly restrictive since we do not require global bounds on the derivatives' size and \eqref{E:gamma_loc_x_bdd} is an extension of the case when $\gamma$ is uniformly bounded.

However, condition $3)$ deserves comment. First of all,  it automatically holds when $\gamma$ is independent of the contract rate $m$. When $\gamma$ does depend upon $m$, that $\gamma_m\geq 0$ is natural since prepayments should rise with the current coupon.  Next, under the given regularity assumptions we have (see \eqref{E:gamma_loc_x_bdd}):
\begin{equation}\label{E:gamma_m_ub_0}
\gamma_m(x,m,z) \leq B_\gamma(n) + L_\gamma(n) m;\qquad x\in D_n; m,z\in [0,n].
\end{equation}
Since $\Xi(mT) = 1/(mT)$ for small $m$ we see that in fact, \eqref{E:gamma_m_bdd} is not restrictive for small $m$.  But, for $m$ large it does imply that $\gamma$ is approximately constant in $m$. Note that for $T=30$ the threshold $mT\leq 2$ is satisfied for $m\leq 6.67\%$.
\end{remark}

With the following assumptions in place we define what it means for $m$ to be a current coupon function:

\begin{definition}\label{D:definition_m}
$m:D\mapsto [0,\infty)$ is a \emph{current coupon function} if \eqref{E:m_goal} holds under the measure $\qprob^x$ for all $x\in D$: i.e.
\begin{equation}\label{E:m_goal_x}
0 = \espalt{x}{\int_0^T p(t,m(x))(m(x)-r_t)e^{-\int_0^t\left(r_u + \gamma(X_u,m(x),m(X_u))\right)du}dt};\qquad x\in D.
\end{equation}
\end{definition}

A current coupon function is a fixed point of a non-linear operator $\mathcal{A}$. To see this, note that $m(x)$ is deterministic and hence we can write \eqref{E:m_goal_x} as
\begin{equation}\label{E:m_goal_op}
m(x) = \mathcal{A}[m](x)\dfn \frac{\espalt{x}{\int_0^T p(t,m(x)) r_t e^{-\int_0^t\left(r_u + \gamma(X_u,m(x),m(X_u))\right)du}dt}}{\espalt{x}{\int_0^T p(t,m(x))e^{-\int_0^t\left(r_u + \gamma(X_u,m(x),m(X_u)\right)du}dt}}.
\end{equation}

The complicating features of the above operator are the non-linearity of $\mathcal{A}$ in $m$, and the joint dependence of $\gamma$ on both $m(x),m(X_t)$. Indeed, the first feature means that it is prohibitively difficult to verify if $\mathcal{A}$ is a contraction, and hence we we will have to appeal to a topological fixed point theorem for existence of solutions. Second, due to the presence of $m(x)$ within the expectation, \textit{a-priori} we do no expect any smoothing of the map $m\mapsto \mathcal{A}[m]$, or that $\mathcal{A}$ possesses the compactness properties necessary to invoke any classical topological fixed point theorem.  However, through a delicate localization argument, fixed points do exist under the current assumptions, as Theorem \ref{T:main_result} now shows. The lengthy proof is given in Appendix \ref{S:proof_of_main} below.

\begin{theorem}\label{T:main_result}

Let Assumptions \ref{A:region} -- \ref{A:gamma} hold.  Then, there exists a strictly positive current coupon function $m$: i.e. \eqref{E:m_goal_x} holds. The function $m$ is locally $\alpha$-H\"{o}lder continuous for all $\alpha\in (0,1)$.

\end{theorem}

\section{Perturbation Analysis}\label{S:perturb}

Theorem \ref{T:main_result} asserts the existence of current coupon function.  However, since our method of proof does not use the contraction principle, we do not know if solutions are unique and do not automatically have a method to compute them. One may certainly try an iterative procedure in \eqref{E:m_goal_op}, starting with an arbitrary function $m_0$ on $D$ and, defining $m_n = \mathcal{A}[m_{n-1}], n=1,2,\dots$, but absent a contraction, it is not clear if this procedure converges.  Thus, in this section, we offer a perturbation analysis where the intensity $\gamma$ is perturbed off of a baseline intensity $\gamma_0$ which only depends upon the factors process $X$.  The goal is to uniquely identify $m$ up to leading orders of the perturbation.  With this identification, we then in the next section provide a numerical approximation to the fixed point and compare its performance.

As a starting point, we present a proposition, similar to \cite[Lemma 2.1]{MR2943181}, which shows that when $\gamma_0 = \gamma_0(X)$ only depends upon the factor process $X$, there is a unique current coupon function.

\begin{proposition}\label{P:gamma_factor}
Let Assumptions \ref{A:region} -- \ref{A:pos_ir} hold. Assume $\gamma(x,m,z) = \gamma_0(x)$: and that $\gamma_0$ satisfies $1)-2)$ in Assumption \ref{A:gamma}\footnote{In fact, $\gamma_0$ need only be locally Lipschitz for the result to go through.}. Then there exists a unique fixed point $m(x)$ solving \eqref{E:m_goal_x}, which in this instance reduces to
\begin{equation}\label{E:m_goal_x_factor}
0 = \espalt{x}{\int_0^T p(t,m(x))(m(x)-r_t)e^{-\int_0^t\left(r_u + \gamma_0(X_u)\right)du}dt}.
\end{equation}
The function $m$ is locally $\alpha$-H\"{o}lder continuous on $D$ for any $\alpha\in (0,1)$.
\end{proposition}

\begin{proof}[Proof of Proposition \ref{P:gamma_factor}]

Fix $x\in D$. For $t\leq T$ define
\begin{equation}\label{E:f_g_F_G_def}
\begin{split}
f(t)&\dfn \espalt{x}{e^{-\int_0^t\left(r_u + \gamma(X_u)\right)du}};\qquad F(t)\dfn \int_0^t f(u)du,\\
g(t)&\dfn \espalt{x}{r_t e^{-\int_0^t\left(r_u + \gamma(X_u)\right)du}};\qquad G(t) \dfn \int_0^t g(u)du.
\end{split}
\end{equation}
Next, define
\begin{equation*}
h(T,m):=e^{mT}\int_0^T \left(1-e^{-m(T-t)}\right)(mf(t)-g(t))\ dt;\qquad T>0,m>0.
\end{equation*}
Note that we will have a solution to \eqref{E:m_goal_x_factor} if for each $x\in D, T>0$ we can find a number $m = m(x)>0$ such that $h(T,m) = 0$. Indeed, this follows by plugging in $p(t,m)$ from \eqref{E:balance_closed_form_P01} and noting that $e^{mT}, 1-e^{-m(T-t)}$ are strictly positive. To find such an $m$, note that  $h(0,m)=0$ and
\begin{equation*}
\dfrac{\partial}{\partial T}h(T,m)=me^{mT}\int_0^T\left(mf(t)-g(t)\right) dt = me^{mT}(mF(T)-G(T)),
\end{equation*}
so that $h(T,m)=\int_0^T me^{mt}(mF(t)-G(t))\ dt$. Now, for $G$ from \eqref{E:f_g_F_G_def} we have
\[\begin{aligned}
G(t)&=\espalt{x}{\int_0^t \left(r_u \pm \gamma(X_u)\right) e^{-\int_0^u(r_v+\gamma(X_v))dv}du}; \\
&=1-\espalt{x}{\int_0^t \gamma(X_u) e^{-\int_0^u(r_v+\gamma(X_v))dv}du}-\mathbb{E}^x\left[e^{-\int_0^t (r_v+\gamma(X_v))dv}\right]; \\
&=H(t)-\dot{F}(t),
\end{aligned}\]
where we have set $H(t)\dfn 1-\espalt{x}{\int_0^t \gamma(X_u) e^{-\int_0^u(r_v+\gamma(X_v))dv}du}$. Since $r> 0$:
\begin{equation}\label{E:temp_1}
H(t) >1-\espalt{x}{\int_0^t \left(r_u+\gamma(X_u)\right)e^{-\int_0^u(r_v+\gamma(X_v))dv}du} = \dot{F}(t) > 0.
\end{equation}
Coming back to $h$ we have
\begin{equation*}
h(T,m)=\int_0^T me^{mt}\left(mF(t)+\dot{F}(t)-H(t)\right)\ dt=m\left(e^{mT}F(T)-\int_0^Te^{mt}H(t) \ dt\right).
\end{equation*}
Hence, $h(T,m)=0$ is equivalent to $F(T)-\int_0^T e^{-m(T-t)}H(t)dt=0$. Using \eqref{E:temp_1} it is clear that, as a function of $m$, the left hand side is strictly increasing, takes the value $F(T)-\int_0^T H(t)dt < 0$ at $0$, and limits to $F(T)>0$ as $m\uparrow\infty$. Thus, there is a unique $m$ so that $h(T,m) = 0$. The statement regarding the regularity of $m$ follows from Theorem \ref{T:main_result} since fixed points are unique in this case.

\end{proof}

Having established existence and uniqueness in the baseline case, we now perform the perturbation analysis. To do so, assume

\begin{assumption}\label{A:perturb}

$\gamma(x,m,z) = \gamma_0(x) + \eps \gamma_1(x,m,z)$ where $\gamma_0$ satisfies parts $1),2)$ of Assumption \ref{A:gamma} and $\gamma_1\in C^2(E)$ is compactly supported with derivatives which are continuously extendable to $D\times \cbra{0}\times\cbra{0}$.

\end{assumption}

Under Assumptions, \ref{A:region} -- \ref{A:pos_ir} and \ref{A:perturb} it follows from Theorem \ref{T:main_result} that for $\eps > 0$ small enough, there exists a continuous current coupon function $m^\eps$. In fact, $m^{\eps}$ is unique up to leading orders of $\eps$ as well as explicitly identifiable, as the following theorem shows:

\begin{theorem}\label{T:perturb}

Let Assumptions \ref{A:region}--\ref{A:pos_ir} and \ref{A:perturb} hold.  For $\eps>0$ small enough, let $m^{\eps}$ be \emph{any} current coupon function, continuous on $D$. Then we have
\begin{equation}\label{E:perturb}
m^{\eps}(x) = m_0(x) + \eps m_1(x) + \textrm{o}(\eps).
\end{equation}
Above, the convergence is locally uniform for $x\in D$.  The function $m_0$ is the unique fixed point from Proposition \ref{P:gamma_factor} and, for $x\in D$
\begin{equation}\label{E:m1}
m_1(x) = \frac{\espalt{x}{\int_0^T \left(m_0(x)-r_t\right)p(t,m_0(x))\left(\int_0^t\gamma_1(X_u,m_0(x),m_0(X_u))du\right) e^{-\int_0^t\left(r_u + \gamma_0(X_u)\right)du}dt}}{\espalt{x}{\int_0^T \left((m_0(x)-r_t)p_m(t,m_0(x)) + p(t,m_0(x))\right)e^{-\int_0^t\left(r_u +\gamma_0(X_u)\right)du}dt}}.
\end{equation}

\end{theorem}

Though the formula for $m_1$ is lengthy, the point of Theorem \ref{T:perturb} is that it is \emph{explicitly identifiable} given $m_0$, the unique fixed point in the baseline case.  Additionally, as will be used in the following section, we point out that the formula for $m_1$ makes perfect sense as long as the relevant random variables and expectations are well defined.  In particular, $\gamma_1$ need not be compactly support and $\gamma_0,\gamma_1$ need not be $C^2$ in order for the above formula to make sense.

\begin{proof}[Proof of Theorem \ref{T:perturb}]

For $\eps > 0$ small enough, let $m^{\eps}(x)$ be \textit{any} continuous solution of \eqref{E:m_goal_x} (or equivalently \eqref{E:m_goal_op}) with $\gamma = \gamma_0 + \eps\gamma_1$. From Theorem \ref{T:main_result} we know at least one such function exists.  First, since $p(t,m)\leq 1, \gamma\geq 0, r\geq 0$ the numerator in \eqref{E:m_goal_op} is bounded above by
\begin{equation}\label{E:num_ub}
\espalt{x}{\int_0^T r_t e^{-\int_0^t r_u du}dt} \leq 1.
\end{equation}
Second, using that $\gamma_1$ is compactly supported (and hence bounded above by some $C_{\gamma_1}$) and Lemma \ref{L:p_prop1} below it follows for any $\eps_0>0$ small enough, the denominator in \eqref{E:m_goal_op} is bounded below by, for $\eps < \eps_0$:
\begin{equation*}
\frac{1}{2}e^{-\eps_0C_{\gamma_1} T}\espalt{x}{\int_0^{T/2} e^{-\int_0^t r_u dt}dt}.
\end{equation*}
As a function of $x$ the above is continuous and strictly positive in $D$, where this latter fact follows from the elliptic Harnack inequality: see \cite[Chapter 4]{MR1326606}. Thus, $m^{\eps}$ is locally bounded on $D$, uniformly in $0<\eps < \eps_0$. Now, recall \eqref{E:m_goal_x}, specified to the current setup:
\begin{equation}\label{E:m_goal_x_e}
\begin{split}
0 &= \espalt{x}{\int_0^T\left(m^{\eps}(x) - r_t\right)p(t,m^{\eps}(x))e^{-\int_0^t\left(r_u + \gamma_0(X_u) + \eps \gamma_1(X_u, m^{\eps}(x),m^{\eps}(X_u))\right)du}dt}.
\end{split}
\end{equation}
We first claim that for each $x\in D$, $\lim_{\eps\downarrow 0}m^{\eps}(x) = m_0(x)$. Indeed, since $m^{\eps}$ is locally bounded in $D$, uniformly in $0<\eps<\eps_0$, it follows for each $x\in D$ that $\cbra{m^{\eps}(x)}_{\eps<\eps_0}$ is uniformly bounded. Let $\eps_n\rightarrow 0$ and assume $m^{\eps_n}(x)\rightarrow \tilde{m}(x)$ for some $\tilde{m}(x)$.  Since $\gamma_1$ is continuous and compactly supported, the dominated convergence theorem yields
\begin{equation*}
0 = \espalt{x}{\int_0^T(\tilde{m}(x)-r_t)p(t,\tilde{m}(x))e^{-\int_0^t\left(r_u + \gamma_0(X_u)\right)du}dt},
\end{equation*}
and so by the uniqueness of $m_0$ from Proposition \ref{P:gamma_factor} we know that $\tilde{m}(x) = m_0(x)$.  Since this works for all subsequences $\eps_n\rightarrow 0$ the convergence result holds. Next, define $\ol{m}$ through
\begin{equation}
m^{\eps}(x)=m_0(x)+\eps\ol{m}(x,\eps);\qquad x\in D,\eps<\eps_0.
\end{equation}
Using Taylor's theorem we have
\begin{equation*}
\begin{split}
m^\eps(x)-r_t &= m_0(x)-r_t + \eps\ol{m}(x,\eps);\\
p(t,m^{\eps}(x)) &= p(t,m_0(x)) + \eps\ol{m}(x,\eps)p_m(t,m_0(x)) + \frac{1}{2}\eps^2\ol{m}(x,\eps)^2p_{mm}(t,\xi(x,\eps));\\
e^{-\eps\int_0^t\gamma_1(X_u,m^{\eps}(x),m^{\eps}(X_u))du} &= 1 - \eps\int_0^t\gamma_1(X_u,m^{\eps}(x),m^{\eps}(X_u))du \\
&\qquad + \frac{1}{2}\eps^2\left(\int_0^t \gamma_1(X_u,m^{\eps}(x),m^{\eps}(X_u))du\right)^2\hat{\xi}(x,\eps,t),
\end{split}
\end{equation*}
where
\begin{equation*}
\vert \xi(x,\eps)\vert \leq \eps\vert \ol{m}(x,\eps)\vert;\qquad 0 \leq \hat{\xi}(x,\eps,t) \leq e^{\eps\int_0^t \gamma_1(X_u,m^{\eps}(x),m^{\eps}(X_u))du}.
\end{equation*}
Plugging these expansions back into \eqref{E:m_goal_x_e} and collecting terms by explicit powers of $\eps$, the zeroth order term is
\begin{equation*}
\espalt{x}{\int_0^T (m_0(x)-r_t)p(t,m_0(x))e^{-\int_0^t(r_u + \gamma_0(X_u))du}dt} = 0,
\end{equation*}
where the equality follows from Proposition \ref{P:gamma_factor}.  The first order (in $\eps$) terms, within the expectation and time integral, are
\begin{equation*}
\begin{split}
&\ol{m}(x,\eps)p(t,m_0(x)) + \ol{m}(x,\eps)(m_0(x)-r_t)p_m(t,m_0(x))\\
&\qquad - (m_0(x)-r_t)p(t,m_0(x))\int_0^t\gamma_1(X_u,m^{\eps}(x),m^{\eps}(X_u))du.
\end{split}
\end{equation*}
Using the given regularity, local boundedness and compactly supported assumptions, all higher order terms together are $O(\eps^2)$, uniformly on compact subsets of $D$. Since the zeroth order term vanishes, we may divide \eqref{E:m_goal_x_e} by $\eps>0$ to obtain
\begin{equation*}
\begin{split}
0&= \ol{m}(x,\eps)\espalt{x}{\int_0^T \left(p(t,m_0(x)) + (m_0(x)-r_t)p_m(t,m_0(x))\right)e^{-\int_0^t\left(r_u + \gamma_0(X_u)\right)du}dt}\\
&\qquad + \espalt{x}{\int_0^T (m_0(x)-r_t)p(t,m_0(x))\int_0^t\gamma_1(X_u,m^{\eps}(x),m^{\eps}(X_u))du\ e^{-\int_0^t\left(r_u + \gamma_0(X_u)\right)du}dt} + \frac{O(\eps^2)}{\eps},
\end{split}
\end{equation*}
which can be re-written as
\begin{equation*}
\begin{split}
\ol{m}(x,\eps) &= \frac{\espalt{x}{\int_0^T (m_0(x)-r_t)p(t,m_0(x))\int_0^t\gamma_1(X_u,m^{\eps}(x),m^{\eps}(X_u))du\ e^{-\int_0^t\left(r_u + \gamma_0(X_u)\right)du}dt} + \frac{O(\eps^2)}{\eps}}{\espalt{x}{\int_0^T \left(p(t,m_0(x)) + (m_0(x)-r_t)p_m(t,m_0(x))\right)e^{-\int_0^t\left(r_u + \gamma_0(X_u)\right)du}dt}};\\
&=m_1(x) + \frac{\espalt{x}{\int_0^T(m_0(x)-r_t)p(t,m_0(x)R(t;x,\eps)e^{-\int_0^t \left(r_u+\gamma_0(X_u)\right)du}dt} + \frac{O(\eps^2)}{\eps}}{\espalt{x}{\int_0^T \left(p(t,m_0(x)) + (m_0(x)-r_t)p_m(t,m_0(x))\right)e^{-\int_0^t\left(r_u + \gamma_0(X_u)\right)du}dt}},
\end{split}
\end{equation*}
where
\begin{equation*}
R(t; x,\eps) \dfn \int_0^t\left(\gamma_1(X_u,m^{\eps}(x),m^{\eps}(X_u))-\gamma_1(X_u,m_0(x),m_0(X_u))\right)du.
\end{equation*}
We have already shown that $m^{\eps}(x)\rightarrow m_0(x)$. Since $m^{\eps}$ is continuous, $m^\eps$ converges to $m_0$ uniformly on compact subsets of $D$.  Since $\gamma_1$ is $C^2$ and compactly supported it thus follows by the dominated convergence theorem that $\lim_{\eps\downarrow 0} \ol{m}(x,\eps)-m_1(x) = 0$ with uniform convergence on compact subsets of $D$, finishing the result.

\end{proof}

\section{A Numerical Approximation}\label{S:numerical}

Theorem \ref{T:perturb} offers a natural numerical approximation for computing current coupon functions. Namely, for a given intensity function $\gamma$ we first identify if there is a decomposition
\begin{equation}\label{E:gamma_decomp}
\gamma(x,m,z) = \gamma_0(x) + \gamma_1(x,m,z),
\end{equation}
and then we compute $m_0$ from $\gamma_0$, define $m_1$ as in \eqref{E:m1} and output the approximation from Theorem \ref{T:perturb} at $\eps = 1$: i.e.
\begin{equation}\label{E:m_approx}
m(x)\approx m_0(x) + m_1(x).
\end{equation}
Note that this approximation is obtainable as long as $m_0,m_1$ are well defined, and does not necessarily require $\gamma_0,\gamma_1$ to satisfy the regularity and growth conditions in Assumption \ref{A:gamma}. Computationally, the advantage of this approximation over naive contraction is clear: there is only one Monte Carlo simulation (for each $x\in D$ along a give mesh) needed to compute $m_1$. 

Next,  we point out that a decomposition \eqref{E:gamma_decomp} is always possible since one may take $\gamma_0 =0$. In this instance, $m_0(x)$ from Proposition \ref{P:gamma_factor} solves
\begin{equation}\label{E:m0_0_gamma}
\frac{1-e^{-m_0(x)T}}{m_0(x)T} = \frac{1}{T}\int_0^T\espalt{\qprob^x}{e^{-\int_0^t r_udu}}dt;\qquad x\in D.
\end{equation}
For many models of interest (e.g. see \cite[Example 6.5.2]{MR2057928} for when $r\sim CIR$), the expectation on the right hand size is explicitly computable and $m_0$ is easily obtained by inverting the strictly decreasing function $y\mapsto (1-e^{-y})/y$. Alternatively, if there is some $\gamma>0$ so that $\gamma(x,m,z)\geq \gamma$ then one can take $\gamma_0(x) = \gamma$ and $\gamma_1(x,m,z) = \gamma(x,m,z) - \gamma$.  Here, for constant $\gamma_0 = \gamma$ calculation shows that $m_0$ satisfies
\begin{equation}\label{E:m0_const_gamma}
\frac{1-e^{-m_0(x)T}}{m_0(x)} = \int_0^Te^{-\gamma t}\espalt{\qprob^x}{e^{-\int_0^t r_udu}}\left(1+\gamma\frac{1-e^{-m_0(x)(T-t)}}{m_0(x)}\right)dt,
\end{equation}
which is easy to obtain numerically given an explicit formula for $\espalt{x}{e^{-\int_0^t r_u}}$. Once $m_0$ is known, one then may compute $m_1$ using Monte Carlo simulation.

\subsection{An Example}\label{SS:numerical_ex}

We now take an example similar to that in \cite[Section 6]{MR2352905} and assume $X$ is a CIR process (i.e. $d=1$, $D=(0,\infty)$ and $X^{(1)} = r$ is a CIR process) and $\gamma$ takes the form
\begin{equation}\label{E:ex_gamma}
\gamma(x,m,z) = \gamma + k(m-z)^+.
\end{equation}
Thus, there is a constant baseline prepayment intensity $\gamma$, and the full intensity is adjusted upwards by the difference between the contract rate $m$ and refinancing rate $z$, when this value is positive.  This adjustment is then scaled by a factor $k > 0$. As in \cite{MR2352905}, we will assume $k=5$ so this is not necessarily a small perturbation off the baseline case.  Here, we perform two approximations. The first sets $\gamma_0(x) = 0, \gamma_1(x) = \gamma + k(m-z)^+$, computes $m_0$ from \eqref{E:m0_0_gamma}, and then $m_1$ from \eqref{E:m1}. The second approximation takes $\gamma_0(x) = \gamma, \gamma_1(x,m,z) = k(m-z)^+$ computes $m_0$ from \eqref{E:m0_const_gamma} and then $m_1$ from \eqref{E:m1}.  For each approximation we compare $m_0 + m_1$ to the '`theoretical fixed point'' $m$ obtained by naive contraction, which in this instance converges rapidly (e.g. after approximately five iterations) to a fixed function for a given initial guess $m^{(0)}$. The model parameters are the same in \cite{MR2352905}: if $dr_t = \kappa(\theta-r_t)dt + \sigma\sqrt{r_t} dW_t$ then $\kappa = 0.25, \theta = 0.06$, $\sigma = 0.1$. Additionally, $\gamma = 0.045$ and $k=5$.

Figure \ref{F:ErrorPlotsGam0} compares $m_0+m_1$ to $m$ when $\gamma_0(x) = 0$.  As shown in the right plot, the approximation does very well, differing by less than $20$ basis points (for an absolute level of $4\%-12\%)$ within the $(2.5\%,97.5\%)$ percentiles of the CIR invariant distribution. In the ``middle'' of the invariant distribution, the approximation is virtually identical to the naive fixed point, with errors consistently between $0-5$ basis points.

Figure \ref{F:ErrorPlotsGam08} makes a similar comparison, using $\gamma_0(x) = \gamma$. Here, the performance is significantly improved with the $(2.5\%,90\%)$ percentiles in that the approximation $m_0 + m_1$ is nearly identical to the function $m$ obtained through niave contraction. Indeed, the difference between $m_0+m_1$ and $m$ is less than $3$ basis points.  However, for large values of $r$ the error is a bit larger than in the previous method, approaching approximately $7$ basis points. 

\begin{figure}
\epsfig{file=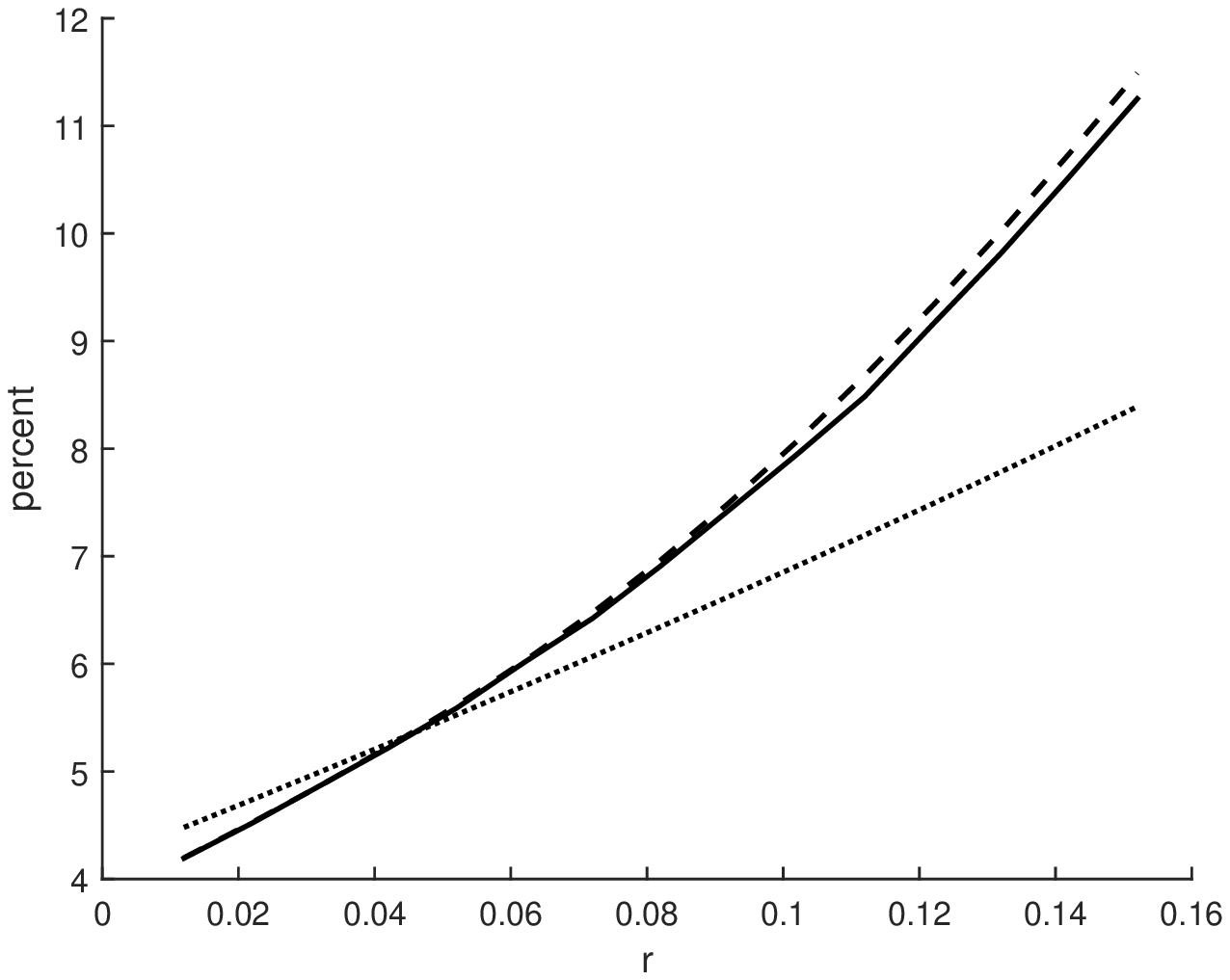,height=4.5cm,width=6.5cm}\quad \epsfig{file=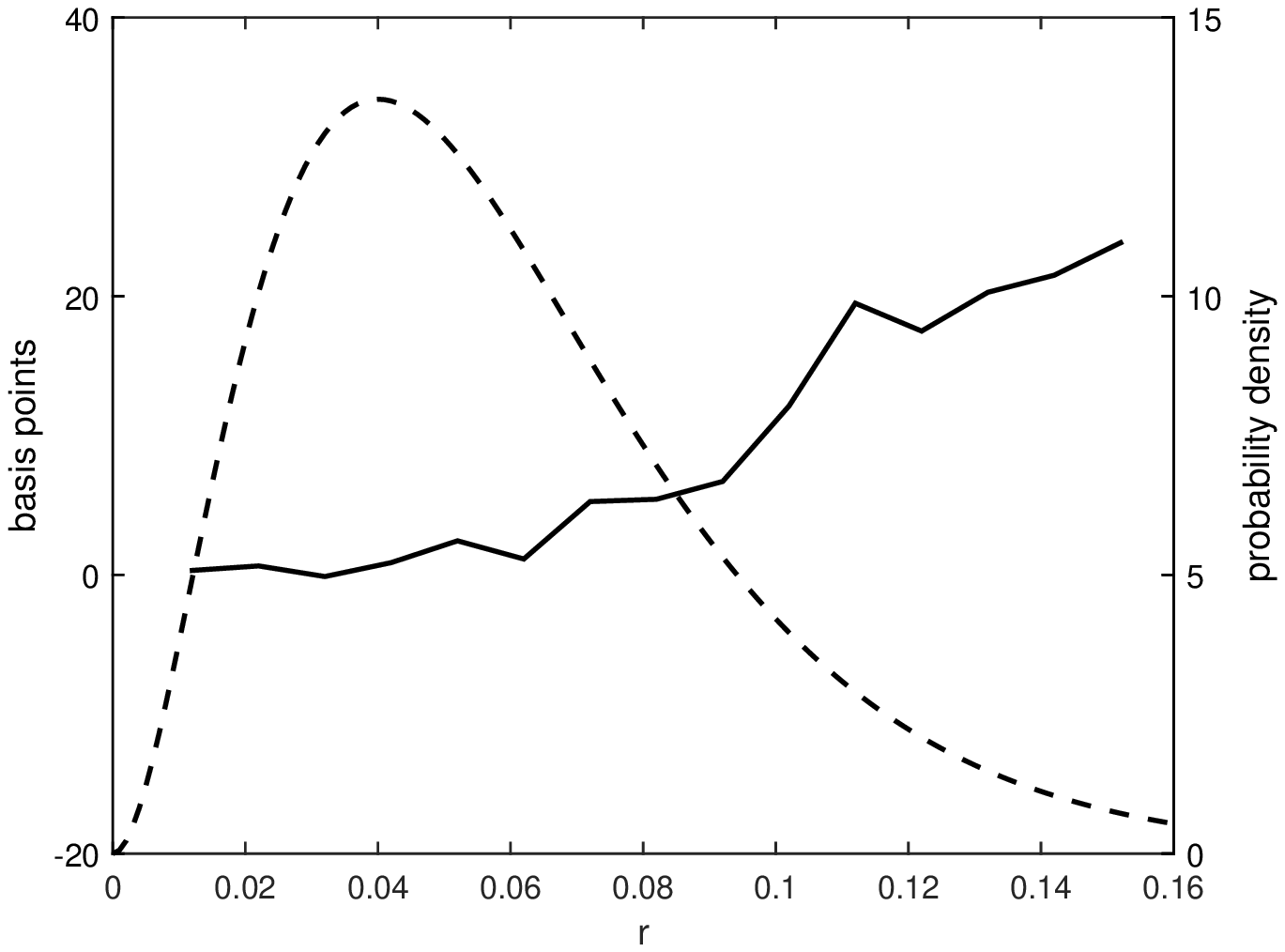,height=4.5cm,width=6.5cm}
\caption{Current coupon functions (left plot) and errors (right plot) as a function of the underlying CIR factor. In the left plot, the solid line is the current coupon function $m$ obtained through naive contraction.  The thick-dash plot is the approximation $m_0 + m_1$ while the thin dash plot is $m_0$.  Values are given in percentage points.  For the right plot, the error is the difference (in basis) points between $m$ and $m_0 + m_1$.  Also in the right plot is the invariant pdf for the CIR process $r$.  $m_0$ is calculated with $\gamma_0(x) = 0$ and $m_1$ is calculated with $\gamma_1(x,m,z) = \gamma + k(m-z)^+$. Parameters are $\kappa = 0.25, \theta = 0.06, \sigma = 0.1$, $T=30$, $k=5$ and $\gamma=0.045$. Computations were performed using \emph{Matlab}, \emph{Mathematica} and the code can be found on the author's website \emph{www.math.cmu.edu/users/scottrob/research}.} \label{F:ErrorPlotsGam0}
\end{figure}

\begin{figure}
\epsfig{file=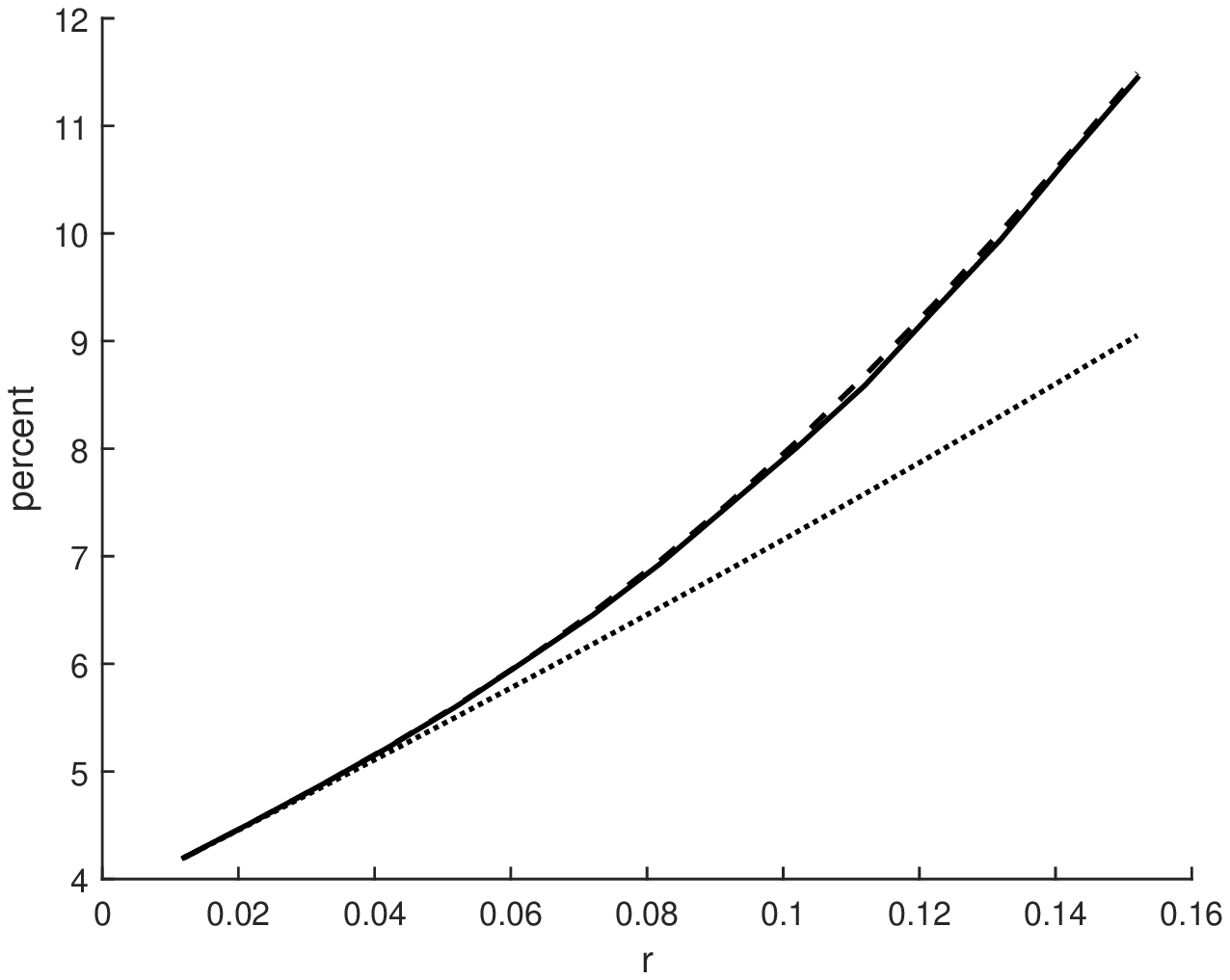,height=4.5cm,width=6.5cm}\quad \epsfig{file=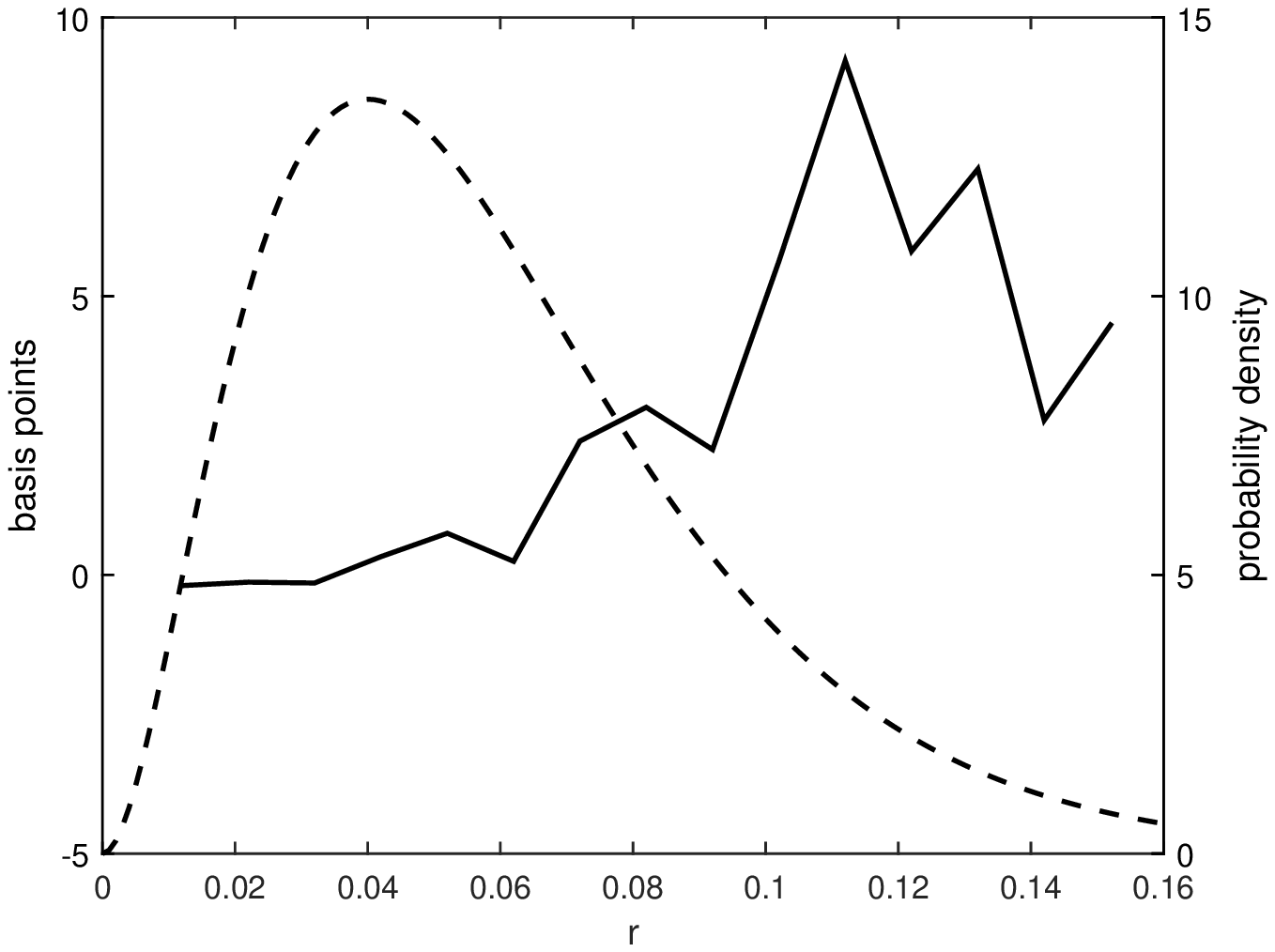,height=4.5cm,width=6.5cm}
\caption{Current coupon functions (left plot) and errors (right plot) as a function of the underlying CIR factor. In the left plot, the solid line is the current coupon function $m$ obtained through naive contraction.  The thick-dash plot is the approximation $m_0 + m_1$ while the thin dash plot is $m_0$.  Values are given in percentage points.  For the right plot, the error is the difference (in basis) points between $m$ and $m_0 + m_1$.  Also in the right plot is the invariant pdf for the CIR process $r$. $m_0$ is calculated with $\gamma_0(x) = \gamma$ and $m_1$ is calculated with $\gamma_1(x,m,z) = k(m-z)^+$.  Parameters are $\kappa = 0.25, \theta = 0.06, \sigma = 0.1$, $T=30$, $k=5$  and $\gamma=0.045$. Computations were performed using \emph{Matlab},\emph{Mathematica} and the code can be found on the author's website \emph{www.math.cmu.edu/users/scottrob/research}.} \label{F:ErrorPlotsGam08}
\end{figure}

\appendix

\section{Proof of Theorem \ref{T:main_result}}\label{S:proof_of_main}

\subsection{Outline of the Proof}\label{SS:proof_outline}

The goal is to show the existence of a function $m:D\mapsto (0,\infty)$ so that \eqref{E:m_goal_op} is satisfied. To do this, we will use Schaefer's Fixed Point Theorem, stated here for the convenience of the reader
\begin{theorem}[Schaefer: \cite{MR1625845}]\label{T:schaefer}
Let $K$ be a closed, convex subset of a Banach space $X$ with $0\in K$. Assume $\mathcal{A}:K\mapsto K$ is continuous, compact and such that $\cbra{u\in K\such u = \lambda \mathcal{A}[u], 0\leq \lambda\leq 1}$ is bounded. Then $\mathcal{A}$ has a fixed point in $K$.
\end{theorem}

It is thus necessary to define the Banach space $X$, closed convex subset $K$ and verify the given assumptions regarding $\mathcal{A}$.  For $X$ we would like to choose the space of $\alpha$-H\"{o}lder continuous functions on $D$ and have $K$ be the subspace of non-negative functions.  However, as $D$ is not necessarily bounded, and the covariance matrix $a$ is not necessarily uniformly elliptic on $D$, we will have a difficult verifying the requisite continuity and compactness of the operator $\mathcal{A}$.  Thus, we must first localize the problem.  At the localized level we will obtain a fixed point using Schaefer's theorem. We will then unwind the localization to get the result. As such, the plan is:
\begin{enumerate}[1)]
\item Define an operator $\mathcal{A}^n$ related to $\mathcal{A}$ and show that $\mathcal{A}^n$ has a fixed point $m^n>0$ defined on $D_n$ which is $\alpha$-H\"{o}lder continuous for all $\alpha\in (0,1)$.
\item For each $m$, obtain uniform (in $n$) H\"{o}lder norm estimates on $D_m$ for the fixed points $m^n,n\geq m+1$.
\item Show that $m^n$ has convergent subsequence with limit $m$ which solves the full fixed point problem.
\end{enumerate}

As a first step in the above plan, we need to obtain \textit{a-prioi} H\"{o}lder norm estimates on solutions to certain partial differential equations (PDE) which are defined through expectations.

\subsection{A Priori Estimates of H\"{o}lder norms}

We first recall the standard definitions of the elliptic and parabolic H\"{o}lder spaces. For a more thorough introduction to such spaces see \cite{MR1814364} for the elliptic case and \cite{MR0181836, lieberman1996second, MR1625845} for the parabolic case.

Fix $n\in\mathbb{N}$ and recall the domain (i.e. open connected region) $D_n$ is bounded with smooth boundary. For $k\in\mathbb{N}$, denote by $C^k(D_n)$ the collection of functions $u$ on $D_n$ such that all partial derivatives of order $\leq k$ are continuous, and by $C^k(\ol{D}_n)$ the subspace of functions with partial derivatives of order $\leq k$ that are continuously extendable to $\partial D_n$. Next, for a given function $u$ on $D_n$ and $\alpha\in (0,1]$ set
\begin{equation*}
\vert u\vert_{D_n} \dfn \sup\limits_{x\in D_n} \vert u(x)\vert;\qquad
[u]_{\alpha,D_n}\dfn \sup\limits_{x,y,\in D_n,x\neq y} \frac{|u(x)-u(y)|}{|x-y|^\alpha}.
\end{equation*}
The space $C^{k,\alpha}(\ol{D}_n)$ is defined as the subset of $C^{k}(\ol{D}_n)$ consisting of those functions $u$, whose partial derivatives of order $\leq k$ have finite $\vert\cdot\vert_{D_n}$ norm and whose partial derivatives of order $k$ have finite $[\cdot]_{\alpha, D_n}$ norm. On the space $C^{k,\alpha}(\ol{D}_n)$ define the norm
\begin{equation}\label{E:spatial_norms}
\Vert u\Vert_{k,\alpha,\ol{D}_n} \dfn \vert u\vert_{D_n} + \sum_{j=1}^k\sup_{|\beta|=j} \vert D^\beta u\vert_{D_n} + \sup_{|\beta|=k} [D^\beta u]_{\alpha,D_n},
\end{equation}
where $\beta$ is a multi-index consisting of $d$ non-negative integers $\beta_1,...,\beta_d$ and $|\beta| = \sum_{i=1}^d \beta_i$ and $D^\beta u = \partial^{|\beta|}_{\beta_1,...,\beta_d}$. It is well known that $C^{k,\alpha}(\ol{D}_n)$ with norm $\Vert\cdot\Vert_{k,\alpha,\ol{D}_n}$ is a Banach space. Lastly, when $k=0$ write $C^{\alpha}(\ol{D}_n)$ for $C^{0,\alpha}(\ol{D}_n)$ and $\Vert\cdot\Vert_{\alpha,\ol{D}_n}$ for $\Vert\cdot\Vert_{0,\alpha,\ol{D}_n}$.

For the parabolic H\"{o}lder norms, define the domain $Q_n:=(0,T) \times D_n$. A typical point $P\in Q_n$ takes the form $P=(t,x), 0<t<T, x\in D_n$. For $P_1=(t,x)$, $P_2=(\bar{t},\bar{x}) \in Q_n$, the parabolic distance between $P_1,P_2$ is $d(P_1,P_2)=(\vert x-\bar{x} \vert ^2 +\vert t-\bar{t} \vert)^{\frac{1}{2}}$. Now, let $\alpha\in (0,1]$. We recall the definitions of standard H\"{o}lder norms of a function $u$ defined on $Q_n$:
\begin{equation}\label{E:norms}
\begin{split}
\vert u \vert_{0,n} &\dfn \sup_{P \in Q_n}\vert u(P) \vert;\qquad \bra{u}_{\alpha,n}\dfn \sup_{P_1,P_2 \in Q_n,P_1\neq P_2}\dfrac{\vert u(P_1)-u(P_2) \vert}{d(P_1,P_2)^{\alpha}}; \\
\vert u \vert_{\alpha,n} &\dfn \vert u \vert_{0,n}+ \bra{u}_{\alpha,n}; \\
\vert u \vert_{2+\alpha,n}&= \vert u \vert_{0,n} + \sum_{i=1}^d \vert D_i u \vert_{0,n} + \sum_{i,j=1}^d\vert D^2_{ij} u\vert_{\alpha,n} + \vert D_t u \vert_{\alpha,n}.
\end{split}
\end{equation}
Above, $D_iu = D^1_{0,...,1,...,0}$ and $D^2_{ij}u = D^2_{0,...,1,...,1,...0}u$ with the ones at $i$ and $i,j$ respectively.

We now prove three lemmas which establish \textit{a priori} estimates (both local and global) for the $\alphanorm{\cdot}{n}$ norm and $\Vert \cdot \Vert_{2,\alpha,\ol{D}_n}$ norm of some conditional expectation expressions, which will be essential in the proofs below.  For each $n$, denote by $\tau_n$ the the first exit time of the process $X$ from $D_n$. Each of the lemmas below concern the function $u:D_n\mapsto\reals$ defined by
\begin{equation}\label{E:u_def}
u(x)\dfn\espalt{x}{\int_0^{T \wedge \tau_n}g(t,X_t)e^{-\int_0^t h(u,X_u)du}dt};\qquad x\in D_n,
\end{equation}
where $g(t,x)$ and $h(t,x)$ are functions defined on $Q_n$. To ease presentation, the bounding constants below may change from line to line, and the $n$ in the constants is assumed to absorb  $K_1(n)$,$K_2(n)$, $B_\gamma(n)$, $L_\gamma(n)$ of Assumptions \ref{A:region}--\ref{A:gamma}, as well as the dimension $d$, parabolic domain $Q_n$, and horizon $T$.  We will keep the dependence upon the H\"{o}lder parameter $\alpha$ explicit.


\begin{lemma}[Global $C^{2,\alpha}$ estimate]\label{L:sim_global_schauder}
Let $u: D_n \mapsto\reals$ be defined in \eqref{E:u_def} and assume for some $\alpha\in (0,1]$, $g$ and $h$ satisfy
\begin{equation*}
\begin{split}
\vert g \vert_{\alpha,n} < \infty;&\qquad\ \vert h \vert_{\alpha,n} \leq K_3(n),\\
\lim_{y\rightarrow x,t\rightarrow T} g(t,y) = 0;&\qquad x\in\partial D_n
\end{split}
\end{equation*}
for some positive constant $K_3(n)$. Then
\begin{equation*}
\Vert u \Vert_{2,\alpha,\ol{D}_n} \leq C(n,K_3(n),\alpha)\cdot \vert g \vert_{\alpha,n}.
\end{equation*}
\end{lemma}


\begin{proof}
Clearly $u(x)=U(0,x)$, where
\begin{equation*}
U(t,x)\dfn\espalt{x}{\int_t^{T \wedge \tau_n}g(s,X_s)e^{-\int_s^t h (\theta,X_\theta)d\theta}\ dt};\qquad t\leq T, x\in D_n.
\end{equation*}
Under the given regularity and ellipticity assumptions, \cite[Theorem 3.7]{MR0181836} implies $U$ is the unique solution to the Cauchy-Dirichlet problem
\begin{equation}\label{E:U_CD_prob}
\begin{cases}
&U_t+\mathcal{L}U-h(t,x)U=-g(t,x), \quad (t,x) \in Q_n, \\
&u(T,x)=0, \quad x \in D_n, \\
&u(t,x)=0, \quad (t,x) \in [0,T]\times \partial D_n.
\end{cases}
\end{equation}
The boundary Schauder estimate (see \cite[Theorems 3.6, 3.7]{MR0181836} and note the condition on $g$ as $t\uparrow T,y\rightarrow x$ is the compatibility condition therein) for parabolic equations yields
\begin{equation*}
\Vert u \Vert_{2,\alpha,\ol{D}_n} \leq \vert U \vert_{2+\alpha,n} \leq C(n,K_3(n),\alpha)\vert g \vert_{\alpha,n}.
\end{equation*}

\end{proof}


\begin{lemma}[Global $C^{\alpha}$ estimate]\label{L:sim_global_holder}
Let $u: D_n \mapsto \reals$ be defined in \eqref{E:u_def} and assume for some $\alpha_0\in (0,1]$ that $g,h$ satisfy
\begin{equation*}
\vert g \vert_{\alpha_0,n} < \infty,\ \vert h \vert_{\alpha_0,n} < \infty,\  \vert h \vert_{0,n} \leq K_4(n),
\end{equation*}
for some positive constant $K_4(n)$. Then for all $\alpha \in (0,1)$
\begin{equation*}
\Vert u \Vert_{\alpha,\ol{D}_n} \leq C(n,K_4(n),\alpha,\alpha_0)\cdot\vert g \vert_{0,n}.
\end{equation*}
\end{lemma}


\begin{proof}
Since $g,h$ are $\alpha_0$-H\"{o}lder continuous, we can invoke \cite[Theorem 5.2]{MR2295424} regarding stochastic representations of solutions to parabolic PDEs to write $u(x) = U(0,x)$ where $U$ satisfies the linear parabolic PDE in \eqref{E:U_CD_prob}. Using the boundary Boundary $W^{2,1}_p$ estimate for parabolic equations in \cite[Theorem 7.3.2]{lieberman1996second} we obtain for all $p>1$,
\begin{equation*}
\Vert U\Vert_{L^{p}(Q_n)}+\Vert DU\Vert_{L^{p}(Q_n)}+\Vert U_t\Vert_{L^{p}(Q_n)} \leq C(n,K_4(n),\alpha_0)\vert g\vert_{0,n}.
\end{equation*}
Now, let $\alpha \in (0,1)$. Since $Q_n$ is a Lipschitz domain  we can apply the Sobolev embedding (Morrey's inequality) to get, for a sufficiently large $p$ depending upon $\alpha$ (as well as the model coefficients, domain, $\alpha_0$, etc.)
\begin{equation*}
\Vert u \Vert_{\alpha,\ol{D}_n} \leq \vert U\vert_{\alpha,n} \leq C(n,K_4(n),\alpha,\alpha_0)\Vert U \Vert_{W^{1,p}(Q_n)} \leq C(n,K_4(n),\alpha,\alpha_0)\vert g\vert_{0,n}.
\end{equation*}
\end{proof}


\begin{lemma}[Interior $C^{\alpha}$ estimate]\label{L:sim_interior_holder}
Let $u:D_n\mapsto\reals$ be defined in \eqref{E:u_def} and assume for some $\alpha_0\in (0,1]$ that $g,h$ satisfy
\begin{equation*}
\vert g \vert_{\alpha_0,n} < \infty,\ \vert h \vert_{\alpha_0,n} < \infty,\  \vert h \vert_{0,n} \leq K_4(n),
\end{equation*}
for some positive constant $K_4(n)$. Let $\alpha\in (0,1)$. We then have for all $m<n$ that
\begin{equation*}
\Vert u \Vert_{\alpha,\ol{D}_m} \leq C(m,K_4(m+1),\alpha,\alpha_0)\cdot\left(\vert g \vert_{0,m+1}+\vert U \vert_{0,m+1}\right),
\end{equation*}
where $U$ is satisfies the linear parabolic PDE \eqref{E:U_CD_prob}.
\end{lemma}

\begin{proof}
Again  $u(x)=U(0,x)$, where $U$ satisfies \eqref{E:U_CD_prob}. Set
\begin{equation*}
Q_m^\prime:=\left(0,\frac{T}{2}\right)\times D_m.
\end{equation*}
For $p\geq 2$, the interior $W^{2,1}_p$ estimate for parabolic equations \cite[Theorem 7.22]{lieberman1996second} yields
\begin{equation*}
\Vert U\Vert_{L^{p}(Q_m^\prime)}+\Vert DU\Vert_{L^{p}(Q_m^\prime)}+\Vert U_t\Vert_{L^{p}(Q_m^\prime)} \leq C(m,K_4(m+1),\alpha_0)\left(\vert g\vert_{0,m+1}+\vert U\vert_{0,m+1}\right).
\end{equation*}
Since $Q_m^\prime$ is a Lipschitz domain, Sobolev embedding yields for any $\alpha\in (0,1)$ by taking $p$ large enough that
\begin{equation*}
\begin{split}
\Vert u \Vert_{\alpha,\ol{D}_m} &\leq \Vert U\Vert_{\alpha, Q^\prime_m} \leq C(m,K_4(m+1),\alpha,\alpha_0)\Vert U \Vert_{W^{1,p}(Q_m^\prime)}\\
&\leq C(m,K_4(m+1),\alpha,\alpha_0)\left(\vert g\vert_{0,m+1}+\vert U\vert_{0,m+1}\right),
\end{split}
\end{equation*}
where above we have set $\|\cdot\|_{\alpha,Q^\prime_m}$ as the $\alpha$-H\"{o}lder norm on the region $Q^\prime_m$.

\end{proof}

\subsection{The localized problem}\label{SS:local}

Throughout this section, Assumptions \ref{A:region}--\ref{A:gamma} are in force. We first seek functions $m = m^n$ on $D_n$ satisfying (compare with \eqref{E:m_goal_x}), for each $x\in D_n$:
\begin{equation}\label{E:n_fp}
\espalt{x}{\int_0^{T\wedge\tau_n}(m(x) - r_t)p(t,m(x))e^{-\int_0^t(r_u+\gamma(X_u,m(x),m(X_u)))du}dt} + \frac{m(x)^2}{n(1-e^{-m(x)T})}=0.
\end{equation}
The second term above is a correction term introduced to establish local regularity of solutions $m$, and will vanish as $n\uparrow\infty$. To establish existence of solutions, let $\alpha\in (0,1)$ and fix a function $\eta\in \Kn$ where
\begin{equation}\label{E:Kn_def}
\Kn\dfn \cbra{\eta \in C^{\alpha}(\overline{D}_n) : \eta\geq 0},
\end{equation}
and look for functions $m = m^{n,\eta}$ solving, for $x\in D_n$:
\begin{equation}\label{E:n_eta_fp}
\espalt{x}{\int_0^{T\wedge\tau_n}(m(x) - r_t)p(t,m(x))e^{-\int_0^t(r_u+\gamma(X_u,m(x),\eta(X_u)))du}dt} + \frac{m(x)^2}{n(1-e^{-m(x)T})}=0.
\end{equation}
I.e. we substitute $\eta(X_t)$ for $m^n(X_t)$ in $\gamma$. Since $\lim_{m\downarrow 0} m^2/(1-e^{-mT}) = 0$ we define the second term above to be $0$ when $m(x) = 0$.   Proposition \ref{P:full_lemma1} below establishes existence and uniqueness of such functions $m^{n,\eta}$.  This defines the map $\mathcal{A}^n[\eta] \dfn m^{n,\eta}$.  Using the \textit{a-prioi} estimates established in the previous section we then verify this map satisfies the hypotheses of Schaefer's theorem \ref{T:schaefer} and hence there is a fixed point $m^n$ satisfying $m^n = \mathcal{A}^n[m^n]$ which is equivalent to $m^n$ solving \eqref{E:n_fp}.

Before proving Proposition \ref{P:full_lemma1} we state two technical lemmas, proved in Appendix \ref{Ap:sup_pf}. First, define
\begin{equation}\label{E:bound_for_r_n}
C^{(1)}_n\dfn \sup\cbra{x^{(1)}\ :\ x\in D_n};\qquad C_n \dfn \sup\cbra{ |x|\ :\ x\in D_n},
\end{equation}
and note that any solution of \eqref{E:n_fp} must \textit{a-priori} satisfy $0\leq m^n(x) < C^{(1)}_n$. Additionally, as in the previous section, the bounding constants below may change from line to line and their dependence on $n$ is understood to absorb the dependence upon the constants $K_1(n),K_2(n), L_\gamma(n), B_\gamma(n)$ of Assumptions \ref{A:factor_coefficients}, \ref{A:gamma}, as well as the region $D_n$, dimension $d$ and maturity $T$.  To state the lemmas, for $\eta\in K_n$ define the function $k^n(m,x;\eta)$ for $x\in D_n,m>0$ by
\begin{equation}\label{E:sim_eqn_0}
k^{n}(m,x;\eta)\dfn \frac{1}{m}\espalt{x}{\int_0^{T \wedge \tau_n} \left(m-r_t\right) \left(1-e^{-m(T-t)}\right)e^{-\int_0^t (r_u + \gamma(X_u,m,\eta(X_u)))du}dt} + \frac{m}{n},
\end{equation}
and note from \eqref{E:balance_closed_form_P01} that \eqref{E:n_eta_fp} holds if for each $x\in D_n$ we can find $m=m(x)=m^{n,\eta}(x)>0$ so that $k^n(m,x;\eta) = 0$. The first technical lemma establishes regularity of $k^n$ in $(x,m)$ for a fixed $\eta$.


\begin{lemma}\label{L:k_n_tech_lemma}
Let $\alpha\in (0,1)$ and $\eta\in\Kn$ and define $k^n$ as in \eqref{E:sim_eqn_0}. Then
\begin{enumerate}[1)]
\item For a fixed $x\in D_n$, $k^n(\cdot,x;\eta)$ is continuously differentiable on $(0,\infty)$. Furthermore, there exists a constant $A(n)$ such that for all $\eta\in\Kn$, $m>0$ and $x\in D_n$:
\begin{equation}\label{E:k_n_m_deriv_bdd}
\frac{1}{n}\leq \partial_{m}k^n(x,m;\eta)\leq A(n).
\end{equation}
\item For a fixed $m>0$, $k^n(m,\cdot;\eta)\in C^{2,\alpha}(\ol{D}_n)$ and there exists a constant $\Lambda(n,\alphanorm{\eta}{n})$ such that for all $0<m\leq C^{(1)}_n$
\begin{equation}\label{E:k_n_x_deriv_bdd}
\Vert k^n(m,\cdot;\eta)\Vert_{2,\alpha,\ol{D}_n} \leq \Lambda(n,\alphanorm{\eta}{n}).
\end{equation}
For $R>0$, $\Lambda(n,\alphanorm{\eta}{n})$ can be made uniform (i.e. depending only upon $n,R$) for $\alphanorm{\eta}{n}\leq R$.

\end{enumerate}
\end{lemma}


The second lemma establishes regularity of $k^n$ with respect to changes in both $m$ and $\eta$.

\begin{lemma}\label{L:k_n_tech_lemma_2}
For $\eta_1,\eta_2\in \mathbb{K}_n$ and $0 < m_1,m_2\leq C^{(1)}_n$ there exists a constant $\Lambda'(n,\alphanorm{\eta_1}{n},\alphanorm{\eta_2}{n})$ so that
\begin{equation}\label{E:k_n_eta_m_reg}
\begin{split}
&\Vert k^n(m_1,\cdot;\eta_1) - k^n(m_2,\cdot;\eta_2)\Vert_{2,\alpha,\ol{D}_n}\\
&\qquad \leq \Lambda'(n,\alphanorm{\eta_1}{n},\alphanorm{\eta_2}{n})\left(\alphanorm{\eta_1-\eta_2}{n} + |m_1-m_2| + \alphanorm{\eta_1-\eta_2}{n}|m_1-m_2|\right).
\end{split}
\end{equation}
and
\begin{equation}\label{E:k_n_eta_m_reg_2}
\begin{split}
&\sup_{x\in D_n}\left| \partial_m k^n(m_1,x;\eta_1) - \partial_m k^n(m_2,x;\eta_2)\right|\\
&\qquad \leq \Lambda'(n,\alphanorm{\eta_1}{n},\alphanorm{\eta_2}{n})\left(|m_1-m_2| + \alphanorm{\eta_1-\eta_2}{n}\right).
\end{split}
\end{equation}
The constant $\Lambda'$ can be made uniform for all $\alphanorm{\eta_1}{n},\alphanorm{\eta_2}{n}\leq R$ for $R>0$.
\end{lemma}


Having established regularity $k^n$ we now present:

\begin{proposition}\label{P:full_lemma1}
For $\alpha\in (0,1)$ and $\eta \in \Kn$, there exists a unique function $m=m^{n,\eta}$ that is strictly positive in $D_n$ and solves \eqref{E:n_eta_fp} in $D_n$. $m^{n,\eta}$ is continuously differentiable in $D_n$ with gradient
\begin{equation}\label{E:m_n_eta_deriv}
\nabla_x m^{n,\eta}(x) = -\frac{\nabla_x k^n(m,x;\eta)}{\partial_m k^n(m,x;\eta)}\bigg|_{m=m^{n,\eta}(x)}.
\end{equation}
Furthermore, $\forall \beta \in (\alpha,1)$, $m$ satisfies the following \textit{a priori} estimate of the $\beta$-H\"{o}lder norm:
\begin{equation}\label{E:m_norm_est_n}
\Vert m^{n,\eta} \Vert_{\beta,\ol{D}_n} \leq C(n,\beta),
\end{equation}
where $C(n,\beta)$ \emph{does not} depend upon $\eta$.
\end{proposition}


\begin{proof}[Proof of Proposition \ref{P:full_lemma1}]
As mentioned above, it suffices for each $x\in D_n$ to find $m=m(x)=m^{n,\eta}(x)$ so that $k^n(m,x;\eta) = 0$. From Lemma \ref{L:k_n_tech_lemma} we know that $k^n$ is strictly increasing in $m$.  Additionally, by the dominated convergence theorem and that $\gamma\geq 0$, $r_t\leq C^{(1)}_n, t\leq \tau_n$ we have
\begin{equation*}
\begin{split}
\lim_{m\downarrow 0} k^n(m,x;\eta) &= -\espalt{x}{\int_0^{T\wedge\tau_n} r_t(T-t)e^{-\int_0^t(r_u + \gamma(X_u,0,\eta(X_u)))du}dt}<0;\\
\lim_{m\uparrow\infty}k^n(m,x;\eta) &= \infty.
\end{split}
\end{equation*}
So for any $x \in D_n$ there exists an unique $m(x)>0$ such that $k^n(m(x),x;\eta)=0$ and this defines the map $m = m^{n,\eta}:D_n\mapsto (0,\infty)$. We next show the \textit{a priori} estimate for the H\"{o}lder norm of $m$ in \eqref{E:m_norm_est_n}. By definition, $\forall x$, $y\in D_n$,
\begin{equation}\label{E:full_lemma_0}
k^n(m(x),x;\eta)=k^n(m(y),y;\eta)=0,
\end{equation}
which implies
\begin{equation}\label{E:full_lemma_1}
k^n(m(y),y;\eta)-k^n(m(x),y;\eta)=k^n(m(x),x;\eta)-k^n(m(x),y;\eta).
\end{equation}
Since $y$ is fixed, the mean value theorem applied to $m\mapsto k^n(m,y;\eta)$ (which is $C^1$ in $m$ from Lemma \ref{L:k_n_tech_lemma}) asserts the existence of $\xi$ between $m(x)$ and $m(y)$ such that
\begin{equation}\label{E:mvt_m}
\partial_m k^n(\xi,y;\eta)\cdot(m(y)-m(x))=k^n(m(x),x;\eta)-k^n(m(x),y;\eta).
\end{equation}
By Lemma \ref{L:k_n_tech_lemma} we thus have
\begin{equation}\label{E:m_beta_ub}
\vert m(x)-m(y) \vert \leq n \vert k^n(m(x),x;\eta)-k^n(m(x),y;\eta)\vert.
\end{equation}
Now, fix $x$ (think of this as a parameter) and note that $k^n(m(x),\cdot;\eta) = u^{m(x),\eta}$ where $u^{m,\eta}$ is defined in \eqref{E:u_m_eta_def} below.  Noting that $m(x)\leq C^{(1)}_n$ it follows from \eqref{E:gm_hm_def}, \eqref{E:gm_bdd}, \eqref{E:hm_bdd} below, as well as $0\leq y^{(1)} + \gamma(y,m(x),\eta(y) \leq C^{(1)}_n + B_\gamma(n)$ on $D_n$ that we may apply Lemma \ref{L:sim_global_holder} to obtain for all $\beta\in (\alpha,1)$ that
\begin{equation*}
\Vert u^{m(x),\eta} \Vert_{\beta,\ol{D}_n} \leq C(n,K_4(n),\beta,\alpha)\sup_{(t,y)\in Q_n}\left|m(x)-y^{(1)}\right|\frac{1-e^{-m(x)(T-t)}}{m(x)} \leq C(n,K_4(n),\beta,\alpha),
\end{equation*}
where the constant $K_4(n)$ does not depend upon $\eta$.  Thus, from \eqref{E:m_beta_ub} we obtain
\begin{equation*}
\vert m(x)-m(y)\vert \leq n\vert k^n(x,m(x);\eta)-k^n(y,m(x);\eta) \vert \leq C(n,K_4(n),\beta,\alpha_0) \vert x-y \vert^\beta.
\end{equation*}
Since it is clear from \eqref{E:n_eta_fp} that $m^{n,\eta}< C^{(1)}_n$, the estimate in \eqref{E:m_norm_est_n} holds. Lastly, \eqref{E:m_n_eta_deriv} follows immediately from the implicit function theorem since Lemmas \ref{L:k_n_tech_lemma}, \ref{L:k_n_tech_lemma_2} imply that for a fixed $\eta\in\mathbb{K}_n$, $k^n(m,x;\eta)$ is $C^1$ in $(0,C^{(1)}_n)\times D_n$.
\end{proof}


In light of Proposition \ref{P:full_lemma1} we define the map $\mathcal{A}^n:\Kn\mapsto\Kn$ by
\begin{equation}\label{E:A_n_map}
\mathcal{A}^n[\eta] = m^{n,\eta};\qquad \eta\in \Kn.
\end{equation}

The following lemma will be needed in the proof of the continuity of the operator $\sn$.

\begin{lemma}\label{L:sim_continuity}
Let $\alpha\in (0,1)$ and  $\eta_1,\eta_2(x) \in \Kn$. Let $m_1 = \mathcal{A}^n[\eta_1]$, $m_2=\mathcal{A}^n[\eta_2]$. Then, there is a constant $\tilde{\Lambda}(n,\alphanorm{\eta_1}{n},\alphanorm{\eta_2}{n})$ which can be bade uniform for $\alphanorm{\eta_1}{n},\alphanorm{\eta_2}{n}\leq R$ such that
\begin{equation*}
\begin{split}
\sup\limits_{x \in D_n} \vert m_1(x)-m_2(x) \vert &\leq \tilde{\Lambda}(n,\alphanorm{\eta_1}{n},\alphanorm{\eta_2}{n})\alphanorm{\eta_1-\eta_2}{n},\\
\sup\limits_{x \in D_n} \left\vert \nabla_{x}k^n(x,m_1(x);\eta_1)-\nabla_xk^n(x,m_2(x);\eta_2) \right \vert &\leq \tilde{\Lambda}(n,\alphanorm{\eta_1}{n},\alphanorm{\eta_2}{n})\alphanorm{\eta_1-\eta_2}{n},\\
\sup\limits_{x \in D_n} \left\vert \partial_m k^n\left(x,m_1(x);\eta_1\right)-\partial_m k^n\left(x,m_2(x);\eta_2 \right) \right \vert &\leq \tilde{\Lambda}(n,\alphanorm{\eta_1}{n},\alphanorm{\eta_2}{n})\alphanorm{\eta_1-\eta_2}{n}.
\end{split}
\end{equation*}
\end{lemma}


\begin{proof}[Proof of Lemma \ref{L:sim_continuity}]
By definition of $m_1,m_2$ we have for all $x\in D_n$ that $0 = k^n(m_1(x),x;\eta_1) = k^n(m_2(x),x;\eta_2)$ and hence
\begin{equation*}
k^n(m_2(x),x;\eta_2) - k^n(m_1(x),x;\eta_2) = k^n(m_1(x),x;\eta_1) - k^n(m_1(x),x;\eta_2).
\end{equation*}
By the mean value theorem applied to the map $m\mapsto k^n(m,x;\eta_2)$ (which is $C^1$ from Lemma \ref{L:k_n_tech_lemma}) there is some $\xi$ between $m_1(x),m_2(x)$ so that $\partial_m k^n(\xi,x;\eta_2)(m_2(x)-m_1(x)) = k^n(m_2(x),x;\eta_2) - k^n(m_1(x),x;\eta_2)$.  It thus follows that
\begin{equation*}
\begin{split}
\vert m_2(x)-m_1(x)\vert &= \frac{\vert k^n(m_1(x),x;\eta_2)-k^n(m_1(x),x;\eta_1)\vert}{\vert \partial_m k^n(\xi,x;\eta_2)\vert},\\
&\leq n \Lambda'\left(n,\alphanorm{\eta_1}{n},\alphanorm{\eta_2}{n}\right)\alphanorm{\eta_1-\eta_2}{n},\\
&= \tilde{\Lambda}(n,\alphanorm{\eta_1}{n},\alphanorm{\eta_2}{n})\alphanorm{\eta_1-\eta_2}{n}.
\end{split}
\end{equation*}
where the inequality follows from \eqref{E:k_n_eta_m_reg} in Lemma \ref{L:k_n_tech_lemma_2} since $0< m_1(x)< C^{(1)}_n$ on $D_n$. The second inequality follows immediately from the first by \eqref{E:k_n_eta_m_reg} of Lemma \ref{L:k_n_tech_lemma_2}. Similarly, the third inequality follows from the first by \eqref{E:k_n_eta_m_reg_2} of Lemma \ref{L:k_n_tech_lemma_2}.
\end{proof}


The following Proposition establishes a fixed point in $\Kn$:

\begin{proposition}\label{P:sim_local_fixed_point}
Let $\alpha\in (0,1)$. There exists $m^n\in \Kn$ that is strictly positive for $x \in D_n$ and solves the fixed point equation $m^n = \mathcal{A}^n[m^n]$ in $D_n$. Equivalently, $m^n$ satisfies \eqref{E:n_fp}. Furthermore, $\forall \beta \in (\alpha,1)$, $m^n$ satisfies the following \textit{a priori} estimate of the $\beta$-H\"{o}lder norm on $D_n$:
\begin{equation*}
\Vert m \Vert_{\beta,\ol{D}_n} \leq C(n, \beta).
\end{equation*}
\end{proposition}


\begin{proof}[Proof of Proposition \ref{P:sim_local_fixed_point}]
The existence of a fixed point $m^n$ will follow from Theorem \ref{T:schaefer} by verifying the steps below.  Here, the Banach space is $X = C^{\alpha}(\ol{D}_n)$, the closed convex subset containing $0$ is $\Kn$ and the operator $\mathcal{A}$ is $\mathcal{A}^n$ from \eqref{E:A_n_map}.
\begin{enumerate}[1)]
\item \emph{The mapping $\mathcal{A}^n:\Kn \mapsto \Kn$ is continuous.} For any $\eta_1,\eta_2\in \Kn$, let $m_1=\mathcal{A}^n[\eta_1]$ and $m_2=\mathcal{A}^n[\eta_2]$. In light of the first part of Lemma \ref{L:sim_continuity}, we need only consider the $\bra{m_1-m_2}_{\alpha,n}$ semi-norm, and clearly, it suffices to show that $\sup_{x\in D_n}|\nabla_x(m_1(x)-m_2(x))| \leq \calphanorm\alphanorm{\eta_1-\eta_2}{n}$. To this end, we have from Proposition \ref{P:full_lemma1} that for $i=1,...,d$ and $x\in D_n$:
\begin{equation*}
\begin{split}
&\partial_{x_i}\left(m_1(x)-m_2(x)\right) = -\Biggr(\frac{\partial_{x_i}k^n\left(m_1(x),x;\eta_1\right)}{\partial_m k^n(m_1(x),x;\eta_1)}-\frac{\partial_{x_i}k^n\left(m_2(x),x;\eta_2\right)}{\partial_m k^n(m_2(x),x;\eta_2)}\Biggr), \\
&\qquad =-\frac{\partial_{x_i}k^n\left(m_1(x),x;\eta_1\right)-\partial_{x_i}k^n\left(m_2(x),x;\eta_2\right)}{\partial_mk^n(m_1(x),x;\eta_1)} \\
&\qquad\qquad +\frac{\partial_{x_i}k^n(m_2(x),x;\eta_2)\times\left(\partial_m k^n\left(m_1(x),x;\eta_1\right)-\partial_m k^n\left(m_2(x),x;\eta_2\right)\right)}{\partial_m k^n(m_1(x),x;\eta_1)\partial_m k^n(m_2(x),x;\eta_2)},
\end{split}
\end{equation*}
and so from Lemmas \ref{L:k_n_tech_lemma}, \ref{L:sim_continuity} we have
\begin{equation*}
\begin{split}
\vert \partial_{x_i}\left(m_1(x)-m_2(x)\right)\vert &\leq n \left\vert \partial_{x_i}k^n\left(m_1(x),x;\eta_1\right)-\partial_{x_i}k^n\left(m_2(x),x;\eta_2\right) \right\vert \\
&\qquad +n^2 \Lambda_n(n,\alphanorm{\eta_2}{n})\left\vert \partial_m k^n\left(m_1(x),x;\eta_1\right)-\partial_m k^n\left(m_2(x),x;\eta_2\right) \right\vert,\\
&\leq \tilde{\Lambda}(n,\alphanorm{\eta_1}{n},\alphanorm{\eta_2}{n})\left(n+n^2\Lambda(n,\alphanorm{\eta_2}{n}\right)\alphanorm{\eta_1-\eta_2}{n},
\end{split}
\end{equation*}
proving continuity.

\item \emph{The mapping $\sn: \Kn \to \Kn$ is compact.} Let us fix some $\beta \in (\alpha,1)$. Given any bounded sequence $\{\eta_i\}_{i\in \mathbb{N}}$ in $\mathbb{K}_n$, Proposition \ref{P:full_lemma1} yields, $\forall i \in \mathbb{N}$,
\begin{equation*}
\Vert \mathcal{A}^n[\eta_i] \Vert_{C^\beta(\overline{D_n})} \leq C(n,\beta).
\end{equation*}
By the standard compact embeddings of H\"{o}lder spaces, there exists a subsequence $\{\sn[\eta_{i_k}]\}_{k \in \mathbb{N}}$ of $\{\sn[\eta_{i}]\}_{i \in \mathbb{N}}$ such that $\{\sn[\eta_{i_k}]\}_{k \in \mathbb{N}}$ converges in $\Vert \cdot \Vert_{C^{\alpha}(\overline{D}_n)}$ norm to some limit in $\mathbb{K}_n$.

\item \emph{The set $\{m \in \mathbb{K}_n \  : \  m=\lambda\sn[m] \  \text{for some } 0 \leq \lambda \leq 1\}$ is bounded.}  Suppose $m \in \mathbb{K}_n$ satisfies $m=\lambda\sn[m]$  for some $0 \leq \lambda \leq 1$. We have from Proposition \ref{P:full_lemma1}
\begin{equation*}
\Vert m \Vert_{C^\alpha(\overline{D_n})} =\lambda \Vert \sn[m] \Vert_{C^\alpha(\overline{D_n})} \leq  C(n,\alpha).
\end{equation*}
\end{enumerate}
Schaefer's Theorem thus asserts that the operator $\sn$ has a fixed point $m^n$ in $\Kn$. By Proposition \ref{P:full_lemma1}, $m^n$ is strictly positive. Moreover, $m^n$ satisfies the following \textit{a priori} estimate of the $\beta$-H\"{o}lder norm on $D_n$:
\begin{equation*}\Vert m \Vert_{C^\beta(\overline{D_n})} \leq C(n,\beta),\ \forall \beta \in (\alpha,1).
\end{equation*}
\end{proof}

\subsection{Global existence of a fixed point.}\label{SS:global}

For an arbitrary $\alpha\in (0,1)$ and $n \in \mathbb{N}$ we now choose $m^n \in \Kn$ such that $m^n$ is a fixed point of the operator $\sn$ in $\Kn$, where $\sn$ is from \eqref{E:A_n_map}. Let us now fix an arbitrary $\nst \in \mathbb{N}$. The following lemma establishes \textit{a priori} estimates for the $\alpha$-H\"{o}lder norms of $\{m^n(x)\}_{n>\nst}$ in $D_{\nst}$. We adopt the notation $\Lambda(\nst)$ to denote some positive constant that changes from line to line and may depend on the dimension $d$, the model coefficients $K_1(\nst+1), K_2(\nst+1)$ from Assumption \ref{A:factor_coefficients}, the local Lipschitz constant $L_\gamma(\nst+1)$ and local bounded constant $B_\gamma(\nst+1)$ from Assumption \ref{A:gamma}, and the time horizon $T$ and domains $D_{\nst}, D_{\nst+1}$.  If additionally, the constant depends upon the H\"{o}lder exponent $\beta$ we will write $\Lambda(\nst,\beta)$ to stress this dependence.  As such when we write $\Lambda(\nst)$ the constant \emph{does not} depend upon $\beta$.


\begin{lemma}\label{L:sim_local_bound}
Let $\beta\in (0,1)$. For any $\nst \in \mathbb{N}$ there exists a positive constant $\Lambda(\nst,\beta)$ such that $\forall n>\nst$, $\Vert m^n \Vert_{C^{\beta}(\overline{D_{\nst}})}\leq \Lambda(\nst,\beta)$.
\end{lemma}


\begin{proof}[Proof of Lemma \ref{L:sim_local_bound}]
Let $\alpha\in (0,\beta)$. Since $m^n$ solves \eqref{E:n_eta_fp} we have, for $m^n(x)>0$, rearranging terms that for all $n\geq \nst+1$ and $x \in D_{\nst}$:
\begin{equation}\label{E:mn_in_nst_bdd}
\begin{split}
m^n(x)&=\dfrac{\espalt{x}{\int_0^{T \wedge \tau_n} r_t p(t,m^n(x))e^{-\int_0^t\left(r_u + \gamma(X_u,m^n(x),m^n(X_u))\right)du}dt}}{\espalt{x}{\int_0^{T \wedge \tau_n} p(t,m^n(x))e^{-\int_0^t\left(r_u + \gamma(X_u,m^n(x),m^n(X_u))\right)du}dt}+\frac{m^n(x)}{n\left(1-e^{-m^n(x)T}\right)}}, \\
&\leq \dfrac{2}{\inf\limits_{x \in D_{\nst}}\espalt{x}{\int_0^{T/2\wedge\tau_{\nst+1}}e^{-\int_0^t (r_u du +C_\gamma(\nst+1))du}dt}}\leq \Lambda(\nst).
\end{split}
\end{equation}
Above, the second inequality has used \eqref{E:num_ub}, Lemma \ref{L:p_prop1} and the elliptic Harnack inequality. We next turn to the $\beta$-H\"{o}lder semi-norm. From \eqref{E:mvt_m}, for all $x,y\in D_{\nst}$ we have
\begin{equation}\label{E:full_local_holder_est}
\vert m^n(x)-m^n(y) \vert = \left\vert \dfrac{k^n(m^n(x),x;m^n)-k^n(m^n(x),y;m^n)}{\partial_m k^n(\xi,y;m^n)}\right\vert,
\end{equation}
where $\xi$ is some number between $m^n(x)$ and $m^n(y)$. From \eqref{E:k_n_tech_lemma_0}, \eqref{E:k_n_tech_lemma_1} and \eqref{E:k_n_m_lb} below, we obtain
\begin{equation*}
\begin{split}
&\frac{\partial k^n}{\partial m}(\xi,y;m^n)\\
&\geq\espalt{\qprob^y}{\int_0^{T\wedge \tau_n}r_te^{-\int_0^t\left(r_u + \gamma(X_u,m^n(x),m^n(X_u))\right)du}\frac{1-e^{-\xi(T-t)}-\xi(T-t)e^{-\xi(T-t)}}{\xi^2}dt}; \\
&\geq\espalt{\qprob^y}{\int_0^{T/2\wedge \tau_{\nst+1}}r_te^{-\int_0^t\left(r_u + \gamma(X_u,m^n(x),m^n(X_u))\right)du}\frac{1-e^{-m(T-t)}-m(T-t)e^{-m(T-t)}}{m^2}\bigg|_{m=m^n(x)\vee m^n(y)}dt}; \\
&\geq\frac{1-e^{-mT/2}-m(T/2)e^{-m(T/2)}}{m^2}\bigg|_{m=m^n(x)\vee m^n(y)}\espalt{\qprob^y}{\int_0^{T/2\wedge \tau_{\nst+1}}r_te^{-\int_0^t\left(r_u + \gamma(X_u,m^n(x),m^n(X_u))\right)du}dt}; \\
&\geq\Lambda(\nst)\espalt{\qprob^y}{\int_0^{T/2\wedge \tau_{\nst+1}}r_t e^{-\int_0^t r_u du}dt}; \\
&\geq\Lambda(\nst).
\end{split}
\end{equation*}
Above, the second and third inequalities follow since  $m\mapsto m^{-2}(1-e^{-m(T-u)}-m(T-u)e^{-m(T-u)})$ is strictly positive and decreasing in $m$. The fourth inequality uses \eqref{E:mn_in_nst_bdd} and that $\gamma(X_u,m^n(x),m^n(X_u))\leq B_\gamma(\nst+1)$ almost surely for $t\leq T/2\wedge\tau_{\nst+1}$. The last inequality follows by taking the infimum of $\espalt{\qprob^y}{\int_0^{T/2\wedge\tau_{\nst+1}} r_t e^{-\int_0^t r_u du}dt}$ over $y\in D_{\nst}$ and noting that by Harnack's inequality this value is strictly positive given $D_{\nst}$ is strictly contained in $D_{\nst+1}$. For the numerator in \eqref{E:full_local_holder_est} we have
\begin{equation*}
k^n(m^n(x),x;m^n)-k^n(m^n(x),y;m^n) = u^{m^n(x),m^n}(x)-u^{m^n(x),m^n}(y),
\end{equation*}
where $u^{m,\eta}$ is from \eqref{E:u_m_eta_def} below.  Note that $u^{m^n(x),m^n}$ is of the form \eqref{E:u_def} with $g = g^{m^n(x)}$ and $h=h^{m^n(x),m^n}$ from \eqref{E:gm_hm_def} below. Specifically, we have
\begin{equation*}
g^{m^n(x)}(t,y) = (m^n(x)-y^{(1)})\frac{1-e^{-m^n(x)(T-t)}}{m^n(x)};\qquad h^{m^n(x),m^n}(y) = y^{(1)} + \gamma(y,m^n(x),m^n(y)).
\end{equation*}
Since $0<m^n(x)< C^{(1)}_n$ we have from \eqref{E:gm_bdd} and \eqref{E:hm_bdd} that the assumptions of Lemma \ref{L:sim_interior_holder} are satisfied (with $\alpha_0 = \alpha$ since $m^n\in C^{\alpha}(\ol{D}_n)$ for the given, arbitrary $\alpha\in (0,\beta)$) and hence for all $\beta\in (0,1)$ by taking $\alpha\in (0,1), \alpha < \beta$:
\begin{equation*}
\begin{split}
\Vert u^{m^n(x),m^n} \Vert_{C^{\beta}(\ol{D}_{\nst})} &\leq \Lambda(\nst,\beta)\left(\vert g^{m^n(x)}\vert_{0,\nst+1} + \vert u^{m^n(x),m^n}\vert_{0,\ol{D}_{\nst + 1}}\right)\\
&\leq\Lambda(\nst,\beta)\left(\Lambda(\nst+1) + C^{(1)}_{\nst+1} + \vert u^{m^n(x),m^n}\vert_{0,\ol{D}_{\nst+1}}\right).
\end{split}
\end{equation*}
Now, for $y\in D_{\nst}$:
\begin{equation*}
\begin{split}
\vert u^{m^n(x),m^n}(y)\vert &= \vert k^n(m^n(x),y;m^n \vert\\
&\leq \int_0^T \left(1-e^{-m^n(x)(T-t)}\right)\espalt{\qprob^y}{1_{t\leq \tau_n}e^{-\int_0^t\left(r_u+\gamma(X_u,m^n(x),m^n(X_u))\right)du}}dt\\
&\qquad + \espalt{\qprob^y}{\int_0^{T\wedge\tau_n}r_t\frac{1-e^{-m^n(x)(T-t)}}{m^n(x)}e^{-\int_0^t\left(r_u + \gamma(X_u,m^n(x),m^n(X_u))\right)du}dt} + \frac{m^n(x)}{n}\\
&\leq T + T\espalt{\qprob^y}{\int_0^{T\wedge\tau_n}r_te^{-\int_0^t r_udu}dt} + \frac{\Lambda(\nst+1)}{n}\\
&\leq 2T + \frac{\Lambda(\nst+1)}{\nst}\\
&=\Lambda(\nst+1).
\end{split}
\end{equation*}
Hence we conclude that $\vert u^{m^n(x),m^n}\vert_{0,\ol{D}_{\nst+1}}\leq \Lambda(\nst,\beta)$ and thus
\begin{equation*}
\vert k^n(m^n(x),x;m^n)-k^n(m^n(x),y;m^n) \vert \leq \Lambda(\nst,\beta)\vert x-y \vert^\beta.
\end{equation*}
Putting these two estimates together in \eqref{E:full_local_holder_est} gives
\begin{equation*}
\vert m^n(x)-m^n(y)\vert \leq  \Lambda(\nst,\beta)\vert x-y \vert^\beta,\ \forall x,\ y\in D_{\nst},
\end{equation*}
finishing the proof, in view of \eqref{E:mn_in_nst_bdd}.
\end{proof}

With all these preparations, we are now ready to prove Theorem \ref{T:main_result}.

\begin{proof}[Proof of Theorem \ref{T:main_result}]
Note that \eqref{E:m_goal_x} is equivalent to
\begin{equation*}
m(x)=
\dfrac{\espalt{x}{\int_0^{T} r_tp(t,m(x))e^{-\int_0^t\left(r_u +\gamma(X_u,m(x),m(X_u))\right)du}dt}}{\espalt{x}{\int_0^{T} p(t,m(x))e^{-\int_0^t\left(r_u +\gamma(X_u,m(x),m(X_u))\right)du}dt}};\qquad x\in D.
\end{equation*}
Let $\alpha\in (0,1)$. From Lemma \ref{L:sim_local_bound}, there exists a positive constant $\Lambda(1,\alpha)$ such that $\forall n>1$, we have $\Vert m^n \Vert_{\alpha,\ol{D}_1}\leq \Lambda(1,\alpha)$. The Arzel\`{a}-Ascoli theorem asserts the existence of a subsequence of $\{m^n(x)\}_{n>1}$, which we denote by  $\cbra{m^{n_k^{(1)}}(x)}_{k \in \mathbb{N}}$, and some $m^{(1)} \in \mathbb{K}_1$ such that for each $n_k^{(1)}$, $m^{n_k^{(1)}}$ satisfies the equality in \eqref{E:mn_in_nst_bdd} for $x\in D_1$ and such that  $m^{n_k^{(1)}}(x)$ converge to $m^{(1)}(x)$ uniformly in $D_1$ as $k \to \infty$, with $\Vert m^{(1)} \Vert_{\alpha,\ol{D}_1}\leq \Lambda(1,\alpha).$

Applying Lemma \ref{L:sim_local_bound} again, we have that there exists a positive constant $\Lambda(2,\alpha)$ such that $\forall n_k^{(1)}>2$, we have $\Vert m^{n_k^{(1)}} \Vert_{\alpha,\ol{D}_2}\leq \Lambda(2,\alpha)$. The Arzel\`{a}-Ascoli theorem again assures the existence of a subsequence of $\cbra{m^{n_k^{(2)}}(x)}_{k\in\mathbb{N}}$ and some $m^{(2)} \in \mathbb{K}_2$ such that $m^{n_k^{(2)}}$ converge to $m^{(2)}$ uniformly in $D_2$ as $k \to \infty$, with $\Vert m^{(2)} \Vert_{\alpha,\ol{D}_2}\leq \Lambda(2,\alpha).$ Note that by construction, $m^{(2)}(x)= m^{(1)}(x)$ for $x\in D_1$.

The above procedure can be carried out iteratively and we conclude that $\forall l \in \mathbb{N}$, there exists a subsequence of $\{m^{n_k^{(l)}}\}_{k>1}$, denoted by  $\{m^{n_k^{(l+1)}}\}_{k \in \mathbb{N}}$, and function $m^{(l+1)} \in \mathbb{K}_{l+1}$, such that $m^{n_k^{(l+1)}}$ converge to $m^{(l+1)}$ uniformly in $D_{l+1}$ as $k \to \infty$, and $\Vert m^{(l+1)} \Vert_{\alpha,\ol{D}_{l+1}}\leq \Lambda(l+1,\alpha)$. Moreover, by construction,  $m^{(l+1)}(x) =  m^{(l)}(x)$ for $x\in D_l$.

Now, for all $x \in D$, there is some $l \in \mathbb{N}$ such that $x \in D_k$, $\forall k \geq l$. We define $m: D \to [0,\infty)$ by
\begin{equation}\label{E:the_m_def}
m(x):=m^{(l)}(x),
\end{equation}
and note that by construction, $m$ is well defined and $m(x) \in C_{loc}^{\alpha}(D)$, $\forall \alpha \in (0,1)$. We claim that $m$ is the desired fixed point. Indeed, fix $l$ and note that for $x\in D_l$ we have that $m(x) = \lim_{k\to\infty} m^{n_k^{(l')}}(x)$ for any $l'\geq l$. Thus, for any $l'\geq l$ we can write, using \eqref{E:mn_in_nst_bdd},
\begin{equation}\label{E:m_frac_eq}
\begin{split}
m(x)&=\dfrac{\lim\limits_{k \to \infty}\espalt{x}{\int_0^{T \wedge \tau_{n_k^{(l')}}} r_t p(t,m^{n_k^{(l')}}(x))e^{-\int_0^t\left(r_u +\gamma(X_u,m^{n_k^{(l')}}(x),m^{n_k^{(l')}}(X_u))\right)du}dt}}{\lim\limits_{k \to \infty}\espalt{x}{\int_0^{T \wedge \tau_{n_k^{(l')}}} p(t,m^{n_k^{(l')}}(x))e^{-\int_0^t\left(r_u+\gamma(X_u,m^{n_k^{(l')}}(x),m^{n_k^{(l')}}(X_u))\right)du}dt} + \frac{m^{n_{k}^{(l')}}(x)}{n_k^{(l')}\left(1-e^{-m^{n_k^{(l')}}(x)}\right)}} \\
&=:\frac{\bold{A}(l')}{\bold{B}(l')},
\end{split}
\end{equation}
where, (recall $x\in D_l$ and $l$ is fixed)
\begin{equation*}
\begin{split}
\bold{A}(l')&=\lim\limits_{k \to \infty}\espalt{x}{\int_0^{T \wedge \tau_{l'}} r_t p(t,m^{n_k^{(l')}}(x))e^{-\int_0^t\left(r_u +\gamma(X_u,m^{n_k^{(l')}}(x),m^{n_k^{(l')}}(X_u))\right)du}} \\
&\qquad +\lim\limits_{k \to \infty}\espalt{x}{\int_{T \wedge \tau_{l'}}^{T \wedge \tau_{n_k^{(l')}}}r_t p(t,m^{n_k^{(l')}}(x))e^{-\int_0^t\left(r_u +\gamma(X_u,m^{n_k^{(l')}}(x),m^{n_k^{(l')}}(X_u))\right)du}dt} \\
&=\espalt{x}{\int_0^{T \wedge \tau_{l'}} r_t p(t,m(x))e^{-\int_0^t\left(r_u +\gamma(X_u,m(x),m(X_u))\right)du}dt} \\
&\qquad+\lim\limits_{k \to \infty}\espalt{x}{\int_{T \wedge \tau_{l'}}^{T \wedge \tau_{n_k^{(l')}}} r_t p(t,m^{n_k^{(l')}}(x))e^{-\int_0^t\left(r_u +\gamma(X_u,m^{n_k^{(l')}}(x),m^{n_k^{(l')}}(X_u))\right)du}dt},
\end{split}
\end{equation*}
The second equality above follows from the bounded convergence theorem since $0\leq p\leq 1$, $0\leq r_t\leq C^{(1)}_{l'}$, $\gamma\geq 0$ and since $m^{n_k^{(l')}}(X_u) \rightarrow m(X_u)$ almost surely for $u\leq \tau_{l'}$, and also, since $l'\geq l$, from $x\in D_l\subset D_{l'}$ so $m^{n_k^{(l')}}(x)\rightarrow m(x)$.  As for the second term we have
\begin{equation*}
\begin{split}
0&\leq \espalt{x}{\int_{T\wedge\tau_{l'}}^{T\wedge \tau_{n_k^{(l')}}} r_t p(t,m^{n_k^{(l')}}(x)) e^{-\int_0^t\left(r_u + \gamma(X_u,m^{n_k^{(l')}}(x),m^{n_k^{(l')}}(X_u))\right)du}dt},\\
&\leq \espalt{x}{\int_{T\wedge\tau_{l'}}^{T} r_t e^{-\int_0^t r_udu}dt}.
\end{split}
\end{equation*}
Taking $l'\uparrow\infty$ and using the non-explosivity of $X$ along with the monotone convergence theorem it thus follows that
\begin{equation*}
\lim_{l'\uparrow\infty} \bold{A}(l') = \espalt{x}{\int_0^{T} r_t p(t,m(x))e^{-\int_0^t\left(r_u +\gamma(X_u,m(x),m(X_u))\right)du}dt}.
\end{equation*}
Repeating the same calculation for $\bold{B}(l')$ and noting the only difference is a) the absence of $r_t$ which is bounded for $t\leq \tau_{l'}$, and b) the fraction $m^{n_k^{(l')}}(x)/(n_k^{(l')}(1-e^{-m^{n_k^{(l')}}(x)})$ which clearly goes away as $k\uparrow\infty$, it similarly follows that for $x\in D_l$:
\begin{equation*}
\lim_{l'\uparrow\infty} \bold{B}(l') = \espalt{x}{\int_0^{T} p(t,m(x))e^{-\int_0^t\left(r_u +\gamma(X_u,m(x),m(X_u))\right)du}dt}.
\end{equation*}
Thus, since $m(x)$ on the left hand side of \eqref{E:m_frac_eq} did not depend upon $l'$ the result follows.

\end{proof}

\section{Supplementary Proofs from Section \ref{SS:local}}\label{Ap:sup_pf}

\begin{proof}[Proof of Lemma \ref{L:k_n_tech_lemma}]

Note that $r_t,\gamma(X_t,m,\eta(X_t)$ are non-negative and uniformly bounded above by $C^{(1)}_n + B_\gamma(n)$ for $t\leq \tau_n$. Additionally, from \eqref{E:gamma_loc_x_bdd} and \eqref{E:gamma_m_bdd} we have that for all $x\in D_n,m,z\geq 0$ that
\begin{equation}\label{E:gamma_m_ub_00}
\gamma_m(x,m,z) \leq \min\cbra{B_\gamma(n) + L_\gamma(n)m, \Xi(mT)} \leq \begin{cases} B_\gamma(n) + L_\gamma(n) & m\leq 1\\ \Xi(T) & m>1\end{cases}\ := \ol{M}(n),
\end{equation}
so that $\gamma_m(X_t,m,\eta(X_t))$ is almost surely bounded above on $t\leq\tau_n$ by a constant depending only upon $n$. It thus follows by the bounded convergence theorem that we may pull the differential operator (with respect to $m$) within the expected value and integral in \eqref{E:sim_eqn_0} to obtain
\begin{equation}\label{E:k_n_tech_lemma_0}
\begin{split}
\partial_m k^n(m,x,T;\eta) &= \espalt{x}{\int_0^{T\wedge\tau_n}\partial_m\left(\left(1-\frac{r_t}{m}\right)\left(1-e^{-m(T-t)}\right)e^{-\int_0^t\left(r_u + \gamma(X_u,m,\eta(X_u))\right)du}\right)dt} + \frac{1}{n}.\\
\end{split}
\end{equation}
By differentiating and collecting terms (again all interchanges of the integral and derivative are allowed given the current hypotheses) we obtain
\begin{equation}\label{E:k_n_tech_lemma_1}
\begin{split}
&e^{\int_0^t\left(r_u +\gamma(X_u,m,\eta(X_u))\right)du} \times \partial_m\left(\left(1-\frac{r_t}{m}\right)\left(1-e^{-m(T-t)}\right)e^{-\int_0^t\left(r_n + \gamma(X_u,m,\eta(X_u))\right)du}\right)\\
&\qquad = r_t\left(\frac{1-e^{-m(T-t)}}{m^2} - \frac{(T-t)e^{-m(T-t)}}{m} + \frac{1-e^{-m(T-t)}}{m}\int_0^t\gamma_m(X_u,m,\eta(X_u))du\right)\\
&\qquad\qquad + (T-t)e^{-m(T-t)} - (1-e^{-m(T-t)})\int_0^t\gamma_m(X_u,m,\eta(X_u))du.
\end{split}
\end{equation}

For all $m>0,t\leq T$ calculation shows
\begin{equation}\label{E:m_t_basics}
0\leq \frac{1-e^{-m(T-t)}}{m^2} - \frac{(T-t)e^{-m(T-t)}}{m} \leq \frac{1}{2}(T-t)^2;\qquad 0\leq \frac{1-e^{-m(T-t)}}{m} \leq (T-t).
\end{equation}
Since $0\leq \gamma_m(x,m,z)\leq \ol{M}(n)$ and $0\leq r_t \leq C^{(1)}_n$ almost surely in $D_n$ it follows that the right hand side of \eqref{E:k_n_tech_lemma_1} is bounded below by
\begin{equation}\label{E:k_n_tech_lemma_2}
(T-t)e^{-m(T-t)} - (1-e^{-m(T-t)})\int_0^t\gamma_m(X_u,m,\eta(X_u))du,
\end{equation}
and from above by
\begin{equation*}
C^{(1)}_n\left(\frac{1}{2}(T-t)^2 + (T-t)t\ol{M}(n)\right) + (T-t).
\end{equation*}
The upper bound in \eqref{E:k_n_m_deriv_bdd} readily follows. As for the lower bound, from \eqref{E:gamma_m_bdd} we have
\begin{equation}\label{E:k_n_m_lb}
\begin{split}
(T-t)e^{-m(T-t)} &- (1-e^{-m(T-t)})\int_0^t\gamma_m(X_u,m,\eta(X_u))du\\
&\geq (T-t)e^{-m(T-t)} - \Xi(mT)t(1-e^{-m(T-t)});\\
&\geq 0.
\end{split}
\end{equation}
To see the third inequality note that (writing $\beta = 1-t/T$ and multiplying numerator and denominator by $T$)
\begin{equation*}
\begin{split}
\Xi(mT) = \inf_{\beta\in (0,1)}\frac{\beta e^{-\beta mT}}{(1-\beta)(1-e^{-\beta mT})} = \inf_{t\in (0,T)}\frac{(T-t)e^{-m(T-t)}}{t(1-e^{-m(T-t)})},
\end{split}
\end{equation*}
It thus follows from \eqref{E:k_n_tech_lemma_1} that almost surely for all $m>0$ and $t\leq T\wedge\tau_n$ that
\begin{equation*}
\partial_m\left(\left(1-\frac{r_t}{m}\right)\left(1-e^{-m(T-t)}\right)e^{-\int_0^t\left(r_n + \gamma(X_u,m,\eta(X_u))\right)du}\right)\geq 0
\end{equation*}
which yields the upper bound in \eqref{E:k_n_m_deriv_bdd}. Lastly, it is evident from \eqref{E:k_n_tech_lemma_1} that the map
\begin{equation*}
m\mapsto \partial_m\left(\left(1-\frac{r_t}{m}\right)\left(1-e^{-m(T-t)}\right)e^{-\int_0^t\left(r_n + \gamma(X_u,m,\eta(X_u))\right)du}\right)
\end{equation*}
is almost surely continuous in $m$ and non-negative with upper bound
\begin{equation*}
C^{(1)}_n\left(\frac{1}{2}T^2 + T^2\ol{M}(n)\right) + T,
\end{equation*}
and hence by the bounded convergence theorem the map $m\mapsto \partial_m k^n(m,x;\eta)$ is continuous and each $m >0$. Turning to \eqref{E:k_n_x_deriv_bdd}, write $k^n(m,\cdot;\eta) = u^{m,\eta}$ where
\begin{equation}\label{E:u_m_eta_def}
\begin{split}
u^{m,\eta}(x)&\dfn \espalt{x}{\int_0^{T \wedge \tau_n}(m-r_t)\frac{1-e^{-m(T-t)}}{m}e^{-\int_0^t\left(r_u + \gamma(X_u,m,\eta(X_u))\right)du}dt};\qquad x\in D_n
\end{split}
\end{equation}
$u^{m,\eta}$ is of the form \eqref{E:u_def} with
\begin{equation}\label{E:gm_hm_def}
\begin{split}
g^{m}(t,x) &\dfn (m-x^{(1)})\frac{1-e^{-m(T-t)}}{m}\\
h^{m,\eta}(t,x) = h^{m,\eta}(x)&\dfn x^{(1)} + \gamma(x,m,\eta(x)).
\end{split}
\end{equation}
Calculation shows for $0<m\leq C^{(1)}_n$ that
\begin{equation}\label{E:gm_bdd}
\lim_{t\uparrow T,y\rightarrow x} g^m(t,y) = 0, x\in \partial D_n;\qquad \vert g^m\vert_{0,n} \leq C^{(1)}_nT;\qquad \bra{g^m}_{\alpha,n} \leq C^{(1)}_n T^{1-\alpha/2} + T(C^{(1)}_n)^{1-\alpha}
\end{equation}
and
\begin{equation}\label{E:hm_bdd}
\begin{split}
\vert h^{m,\eta}\vert_{0,n} &\leq C^{(1)}_n + B_\gamma(n);\\
\bra{h^{m,\eta}}_{\alpha,n} &\leq (C^{(1)}_n)^{1-\alpha} + L_\gamma(n\vee C^{(1)}_n\vee\alphanorm{\eta}{n})\left((2C_n)^{1-\alpha} + \alphanorm{\eta}{n}\right),
\end{split}
\end{equation}
Note that the above can be made uniform for all $\alphanorm{\eta}{n}\leq R$ for any $R>0$. Thus, Lemma \ref{L:sim_global_schauder} yields the upper bound in \eqref{E:k_n_x_deriv_bdd}.

\end{proof}

\begin{proof}[Proof of Lemma \ref{L:k_n_tech_lemma_2}]

We have $k^n(m_1,\cdot;\eta_1) - k^n(m_2,\cdot;\eta_2) = u^{m_1,\eta_1} - u^{m_2,\eta_2}$ where $u^{m,\eta}$ is from \eqref{E:u_m_eta_def}.  For $0<m_1,m_2\leq C^{(1)}_n$, from \eqref{E:gm_bdd}, \eqref{E:hm_bdd} (applied for the respective $m_i,\eta_i$), it follows from Lemma \ref{L:sim_global_schauder} that for $u^{m_i,\eta_i} = U^{m_i,\eta_i}(0,\cdot)$ where $U^{m_i,\eta_i}$ solves the linear parabolic PDE given in \eqref{E:U_CD_prob}. Furthermore, $\ol{\vert U^{m_i,\eta_i}\vert}_{2,\alpha,\ol{D}_n} \leq C(n,\alphanorm{\eta_i}{n})$ where the bounded constant can be made uniform for $\alphanorm{\eta_i}{n}\leq R$.

Define $V\dfn U^{m_1,\eta_1}-U^{m_2,\eta_2}$. Then $V$ solves the linear parabolic PDE
\begin{equation}\label{E:V_CD_prob_1}
\begin{cases}
&V_t+\mathcal{L}V-h^{m_1,\eta_1}V=-\tilde{g}, \quad (t,x) \in Q_n, \\
&V(T,x)=0, \quad x \in D_n, \\
&V(t,x)=0, \quad (t,x) \in [0,T]\times \partial D_n,
\end{cases}
\end{equation}
where we have set (recall \eqref{E:gm_hm_def}):
\begin{equation}\label{E:tilde_g_here}
\begin{split}
\tilde{g}(t,x)&\dfn g^{m_1}(t,x) - g^{m_2}(t,x) + U^{m_2,\eta_2}(t,x)(h^{m_2,\eta_2}-h^{m_1,\eta_1})(x).
\end{split}
\end{equation}
From \eqref{E:hm_bdd} we have that $\ol{\vert h^{m_1,\eta_1}\vert}_{\alpha,n}$ is bounded from above by a constant which only depends upon $n,\alphanorm{\eta_1}{n}$ (which can be made uniform if $\alphanorm{\eta_1}{n}\leq R$). A lengthy, though direct, calculation shows
\begin{equation*}
\begin{split}
\vert g^{m_1}-g^{m_2}\vert_{0,n} &\leq \left(T+ \frac{1}{2}C^{(1)}_nT^2\right)|m_2-m_1|,\\
\vert h^{m_1,\eta_1}-h^{m_2,\eta_2}\vert_{0,n} &\leq L_\gamma(n\vee C^{(1)}_n \vee \alphanorm{\eta_1}{n}\vee \alphanorm{\eta_2}{n})\left(\alphanorm{\eta_2-\eta_1}{n} + |m_2 - m_1|\right).
\end{split}
\end{equation*}
Note the above, again, can be made uniform for $\alphanorm{\eta_i}{n}\leq R$. Lemma \ref{L:long_h_calc} below shows that there is a constant $\tilde{\Lambda}(n,\alphanorm{\eta_1}{n},\alphanorm{\eta_2}{n})$ (uniform for $\alphanorm{\eta_1}{n},\alphanorm{\eta_2}{n}\leq R$) so that
\begin{equation}\label{E:long_h_calc}
\begin{split}
\bra{g^{m_1}-g^{m_2}}_{\alpha,n} &\leq \left((1+2TC^{(1)}_n)T^{1-\alpha/2} + \frac{1}{2}T^2(C^{(1)}_n)^{1-\alpha}\right)|m_2-m_1|,\\
\bra{ h^{m_1,\eta_1}-h^{m_2,\eta_2}}_{\alpha,n} &\leq \tilde{\Lambda}(n,\alphanorm{\eta_1}{n},\alphanorm{\eta_2}{n})\left(|m_1-m_2| + \alphanorm{\eta_2-\eta_1}{n} + |m_1-m_2|\alphanorm{\eta_2-\eta_1}{n}\right).
\end{split}
\end{equation}
From \eqref{E:tilde_g_here}, it easily follows since $\ol{\vert U^{m_2,\eta_2}\vert}_{2,\alpha,\ol{D}_n}\leq C(n,\alphanorm{\eta_2}{n})$ that (by potentially enlarging $\Lambda''$)
\begin{equation*}
\vert \tilde{g}\vert_{\alpha,n} \leq \tilde{\Lambda}(n,\alphanorm{\eta_1}{n},\alphanorm{\eta_2})\left(|m_1-m_2| + \alphanorm{\eta_1-\eta_2}{n} + |m_1-m_2|\alphanorm{\eta_1-\eta_2}{n}\right).
\end{equation*}
The result then follows from Lemma \ref{L:sim_global_schauder} since $g^m$ and $U^{m_2,\eta_2}$ take the value zero on $t=T,x\in\partial D_n$, and hence the compatibility condition holds.

We next prove \eqref{E:k_n_eta_m_reg_2}.  As follows from \eqref{E:k_n_tech_lemma_0} and \eqref{E:k_n_tech_lemma_1} we have
\begin{equation}\label{E:partial_k_n_m_compare}
\begin{split}
&\partial_m k^n(m_1,x;\eta_1) - \partial_m k^n(k_2,x;\eta_2)\\
&= \espalt{x}{\int_0^{T\wedge\tau_n}\left(A_1(t)\left(B(t)C_1(t) + D_1(t)\right) - A_2(t)\left(B(t)C_2(t) + D_2(t)\right)\right)dt},
\end{split}
\end{equation}
where for $i=1,2$
\begin{equation*}
\begin{split}
A_i(t) &= e^{-\int_0^t\left(r_u + \gamma(X_u,m_i,\eta_i(X_u))\right)du};\qquad B(t) = r_t,\\
C_i(t) &= \frac{1}{m_i^2}\left(1-e^{-m_i(T-t)}-m_i(T-t)e^{-m_i(T-t)}\right) + \frac{1-e^{-m_i(T-t)}}{m_i}\int_0^t\gamma_m(X_u,m_i,\eta_i(X_u))du,\\
D_i(t)) &= (T-t)e^{-m_i(T-t)} - (1-e^{-m_i(T-t)})\int_0^t\gamma_m(X_u,m_i,\eta_i(X_u))du.
\end{split}
\end{equation*}
Using the elementary estimate
\begin{equation*}
\left|A_1(BC_1+D_1) - A_2(BC_2+D_2)\right| \leq |A_1||B||C_1-C_2| + (|B||C_2| + |D_2|)|A_1-A_2| + |A_1||D_1-D_2|,
\end{equation*}
we will obtain the upper bound in \eqref{E:k_n_eta_m_reg_2}.  First, we have the almost sure inequalities
\begin{equation*}
\begin{split}
|A_1(t)| &\leq 1;\qquad |B(t)| \leq C^{(1)}_n,\\
|C_2(t)| &\leq T^2\left(\frac{1}{2} + \ol{M}(n)\right);\qquad |D_2(t)| \leq T(1 + \ol{M}(n)).
\end{split}
\end{equation*}
Above, we have used that $\gamma\geq 0$, $0\leq r_t \leq C^{(1)}_n$ on $t\leq \tau_n$, \eqref{E:m_t_basics}, and \eqref{E:gamma_m_ub_00}.  Next, we have
\begin{equation*}
\begin{split}
&\left|C_1(t)-C_2(t)\right|\\
&\qquad \leq \left|\frac{1-e^{-m_1(T-t)}-m_1(T-t)e^{-m_1(T-t)}}{m_1^2} - \frac{1-e^{-m_2(T-t)}-m_2(T-t)e^{-m_2(T-t)}}{m_2^2}\right|\\
&\qquad \qquad + \frac{1-e^{-m_1(T-t)}}{m_1}\int_0^T\left|\gamma_m(X_u,m_1,\eta_1(X_u))-\gamma_m(X_u,m_2,\eta_2(X_u))\right|du\\
&\qquad\qquad + \int_0^T\gamma_m(X_u,m_2,\eta_2(X_u))du\left|\frac{1-e^{-m_1(T-t)}}{m_1} - \frac{1-e^{-m_2(T-t)}}{m_2}\right|.
\end{split}
\end{equation*}
The map $m\mapsto (1-e^{-m(T-t)}-m(T-t)e^{-m(T-t)})/m^2$ has derivative $-(2/m^3)(1-e^{-m(T-t)}-m(T-t)e^{-m(T-t)} - (1/2)m^2(T-t)^2e^{-m(T-t)})$ which is non-positive and is bounded in absolute value of $(T-t)^3/3 \leq T^3/3$.  Thus,
\begin{equation*}
\left|\frac{1-e^{-m_1(T-t)}-m_1(T-t)e^{-m_1(T-t)}}{m_1^2} - \frac{1-e^{-m_2(T-t)}-m_2(T-t)e^{-m_2(T-t)}}{m_2^2}\right| \leq \frac{T^3}{3}|m_1-m_2|.
\end{equation*}
For the second term we have
\begin{equation*}
\begin{split}
&\frac{1-e^{-m_1(T-t)}}{m_1}\int_0^T\left|\gamma_m(X_u,m_1,\eta_1(X_u))-\gamma_m(X_u,m_2,\eta_2(X_u))\right|du\\
&\qquad \leq T^2 L_\gamma(n\vee C^{(1)}_n\vee\alphanorm{\eta_1}{n}\vee\alphanorm{\eta_2}{n})\left(|m_1-m_2| + \alphanorm{\eta_1-\eta_2}{n}\right).
\end{split}
\end{equation*}
For the third term we have
\begin{equation*}
\begin{split}
&\int_0^T\gamma_m(X_u,m_2,\eta_2(X_u))du\left|\frac{1-e^{-m_1(T-t)}}{m_1} - \frac{1-e^{-m_2(T-t)}}{m_2}\right|\\
&\qquad \leq \frac{1}{2}T^3\ol{M}(n)|m_1-m_2|,
\end{split}
\end{equation*}
since $m\mapsto (1-e^{-m(T-t)})/m$ has a derivative bounded by $(T-t)^2/2$. Thus, we can find a constant $\calphanorm$ so that almost surely for $t\leq T$
\begin{equation*}
|C_1(t)-C_2(t)| \leq \calphanorm\left(|m_1-m_2| + \alphanorm{\eta_1-\eta_2}{n}\right).
\end{equation*}
We next have, by the non-negativity of $r,\gamma$ and the fact that $|e^{-a}-e^{-b}| \leq |a-b|$ for $a,b\geq 0$, that almost surely for $t\leq T\wedge\tau_n$:
\begin{equation*}
\begin{split}
\left|A_1(t) - A_2(t)\right| &\leq \int_0^T\left|\gamma(X_u,m_1,\eta_1(X_u))-\gamma(X_u,m_2,\eta_2(X_u))\right|du,\\
&\leq T L_\gamma(n\vee C^{(1)}_n\vee \alphanorm{\eta_1}{n}\vee\alphanorm{\eta_2}{n})\left(|m_1-m_2| + \alphanorm{\eta_1-\eta_2}{n}\right),\\
&=\calphanorm\left(|m_1-m_2| + \alphanorm{\eta_1-\eta_2}{m}\right).
\end{split}
\end{equation*}
Lastly, we have
\begin{equation*}
\begin{split}
&|D_1(t)-D_2(t)|\\
&\qquad \leq  T\left|e^{-m_1(T-t)}-e^{-m_2(T-t)}\right| + (1-e^{-m_2(T-t)})\int_0^T\left|\gamma_m(X_u,m_1,\eta_1(X_u))-\gamma_m(X_u,m_2,\eta_2(X_u))\right|du\\
&\qquad\qquad + \int_0^T\gamma_m(X_u,m_2,\eta_2(X_u))du\left|e^{-m_2(T-t)}-e^{-m_1(T-t)}\right|,\\
&\qquad \leq T^2|m_1-m_2| + TL_\gamma(n\vee C^{(1)}_n\vee\alphanorm{\eta_1}{n}\vee\alphanorm{\eta_2}{n})\left(|m_1-m_2| + \alphanorm{\eta_1-\eta_2}{n}\right)\\
&\qquad\qquad + \ol{M}(n)T^2|m_1-m_2|,\\
&\qquad \leq \calphanorm \left(|m_1-m_2| + \alphanorm{\eta_1-\eta_2}{n}\right).
\end{split}
\end{equation*}
Putting this all together in \eqref{E:partial_k_n_m_compare} gives for all $x\in D_n$ that
\begin{equation*}
\left|\partial_m k^n(m_1,x;\eta_1) - \partial_m k^n(m_2,x;\eta_2)\right| \leq \calphanorm\left(|m_1-m_2| + \alphanorm{\eta_1-\eta_2}{n}\right),
\end{equation*}
which is the desired result.
\end{proof}


\begin{lemma}\label{L:long_h_calc}

For $0<m_1,m_2\leq C^{(1)}_n$, $\eta_1,\eta_2\in\mathbb{K}_n$ and $g^m,h^m$ as in \eqref{E:gm_hm_def} the inequalities in \eqref{E:long_h_calc} hold.

\end{lemma}


\begin{proof}

The proof is a lengthy calculation based off of Taylor's formula, using the fact that $\gamma$ is both $C^2$, with derivatives of order $\leq 2$ which can be continuously extended to $D\times\cbra{0}\times\cbra{0}$, as well as such that all derivatives of order $\leq 2$ are Lipschitz continuous in $\bar{D}_n\times [0,n]\times [0,n]$ with Lipschitz constant $L_\gamma(n)$.  In particular, for any partial derivative $u$ of $\gamma$ with order $\leq 2$, any $n$ and constants $m_n,z_n>0$
\begin{equation*}
\begin{split}
\sup_{x\in D_n,m\leq m_n,z\leq z_n} |u(x,m,z)| &< \infty,\\
\sup_{x,x'\in D_n;m,m'\leq m_n; z,z'\leq z_n} |u(x,m,z)-u(x',m',z')| &\leq L_\gamma(n\vee m_n\vee z_n)\left(|x-x'| + |m-m'| + |z-z'|\right).
\end{split}
\end{equation*}
The above inequalities are used repeatedly in the sequel. Also, $C(n,\alphanorm{\eta_1}{n},\alphanorm{\eta_2}{n})$ is a constant which may change from line to line and can always be made uniform in $\eta_1,\eta_2$ for $\alphanorm{\eta_1}{n},\alphanorm{\eta_2}{n}\leq R$. Now, for $s,t<T,x,y\in D_n$ we have
\begin{equation*}
\begin{split}
&g^{m_1}(t,x)-g^{m_2}(t,x) - \left(g^{m_1}(s,y) - g^{m_2}(s,y)\right)\\
&= (m_1-x^{(1)})\frac{1-e^{-m_1(T-t)}}{m_1} - (m_2 - x^{(1)})\frac{1-e^{-m_2(T-t)}}{m_2}\\
&\qquad -\left( (m_1-y^{(1)})\frac{1-e^{-m_1(T-s)}}{m_1} - (m_2-y^{(1)})\frac{1-e^{-m_2(T-s)}}{m_2}\right),\\
&= \int_{m_2}^{m_1}\left((T-t)e^{-m(T-t)}+\frac{x^{(1)}}{m^2}\left(1-e^{-m(T-t)}-m(T-t)e^{-m(T-t)}\right)\right)dm\\
&\qquad -\int_{m_2}^{m_1}\left((T-s)e^{-m(T-s)}+\frac{y^{(1)}}{m^2}\left(1-e^{-m(T-s)}-m(T-s)e^{-m(T-s)}\right)\right)dm.
\end{split}
\end{equation*}
We have
\begin{equation*}
\begin{split}
&\left|\int_{m_2}^{m_1}\left((T-t)e^{-m(T-t)} - (T-s)e^{-m(T-s)}\right)dm\right| = \left|\int_{m_2}^{m_1}\int_s^t e^{-m(T-\tau)}(m(T-\tau)-1)d\tau dm\right|\\
&\qquad \leq (1+C^{(1)}_nT)|t-s||m_1-m_2|.
\end{split}
\end{equation*}
Next, we have
\begin{equation*}
\begin{split}
&\left|\frac{x^{(1)}}{m^2}\left(1-e^{-m(T-t)}-m(T-t)e^{-m(T-t)}\right)-\frac{y^{(1)}}{m^2}\left(1-e^{-m(T-s)}-m(T-s)e^{-m(T-s)}\right)\right|\\
&\qquad \leq x^{(1)}\left|\frac{1-e^{-m(T-t)}-m(T-t)e^{-m(T-t)}}{m^2} - \frac{1-e^{-m(T-t)}-m(T-t)e^{-m(T-t)}}{m^2}\right|\\
&\qquad\qquad + |x^{(1)}-y^{(1)}|\frac{1-e^{-m(T-s)}-m(T-s)e^{-m(T-s)}}{m^2}.
\end{split}
\end{equation*}
For any $k\geq 0$ the function $m\mapsto m^{-2}\left(1-e^{-km}-km e^{-km}\right)$ is non-negative and deceasing in $m>0$ with limit as $m\rightarrow 0$ of $(1/2)k^2$. Using this we have
\begin{equation*}
|x^{(1)}-y^{(1)}|\frac{1-e^{-m(T-s)}-m(T-s)e^{-m(T-s)}}{m^2} \leq \frac{1}{2}(T-s)^2|x^{(1)}-y^{(1)}|\leq \frac{T^2}{2}|x^{(1)}-y^{(1)}|.
\end{equation*}
Next, for any $m>0$ the map $m\mapsto m^{-2}\left(1-e^{-m(T-\tau)} - m(T-\tau)e^{-m(T-\tau)}\right)$ has derivative $-(T-\tau)e^{-m(T-\tau)}$ which is bounded above in absolute value on $\tau\leq T$ by $T$.  This implies
\begin{equation*}
x^{(1)}\left|\frac{1-e^{-m(T-t)}-m(T-t)e^{-m(T-t)}}{m^2} - \frac{1-e^{-m(T-t)}-m(T-t)e^{-m(T-t)}}{m^2}\right| \leq C^{(1)}_n T |t-s|.
\end{equation*}
Putting these two terms together gives
\begin{equation*}
\begin{split}
&\left| \int_{m_2}^{m_1}\left(\frac{x^{(1)}}{m^2}\left(1-e^{-m(T-t)}-m(T-t)e^{-m(T-t)}\right) - \frac{y^{(1)}}{m^2}\left(1-e^{-m(T-s)}-m(T-s)e^{-m(T-s)}\right)\right)dm\right|\\
&\qquad \leq \left(\frac{T^2}{2}|x^{(1)}-y^{(1)}| + C^{(1)}_nT|t-s|\right)|m_1-m_2|.
\end{split}
\end{equation*}
Therefore
\begin{equation*}
\begin{split}
&\left|g^{m_1}(t,x)-g^{m_2}(t,x) - \left(g^{m_1}(s,y) - g^{m_2}(s,y)\right)\right|\\
&\qquad \leq |m_1-m_2|\left((1+2C^{(1)}_nT)|t-s| + \frac{T^2}{2}|x^{(1)}-y^{(1)}|\right),
\end{split}
\end{equation*}
and hence
\begin{equation*}
\bra{g^{m_1}-g^{m_2}}_{\alpha,n}\leq |m_1-m_2|\left((1+2C^{(1)}_nT)T^{1-\alpha/2} + \frac{T^2}{2}(C^{(1)}_n)^{1-\alpha}\right),
\end{equation*}
which is \eqref{E:long_h_calc} for $g$. Turning to $h$, write $\bold{a}_i(x) \dfn (x,m_i,\eta_i(x))$ for $i=1,2$ and $x\in D_n$. Set
\begin{equation}\label{E:bold_a_big_bound}
M_n \dfn n\vee C^{(1)}_n \vee \alphanorm{\eta_1}{n}\vee \alphanorm{\eta_2}{n},
\end{equation}
and note that
\begin{equation}\label{E:bold_a_region}
\bold{a}_i(x) \in \bar{E}_{M_n} = \bar{D}_{M_n}\times \bra{0,M_n}\times \bra{0,M_n};\qquad x\in D_n.
\end{equation}
We have, from the second order Taylor formula
\begin{equation}\label{E:h_taylor1}
\begin{split}
&h^{m_1,\eta_1}(x) - h^{m_2,\eta_2}(x) - \left(h^{m_1,\eta_1}(y) - h^{m_1,\eta_1}(y)\right)\\
&\ =\gamma(\bold{a}_1(x)) - \gamma(\bold{a}_2(x)) -\left(\gamma(\bold{a}_1(y)) - \gamma(\bold{a}_2(y))\right),\\
&\ = (m_1-m_2)\left(\gamma_m(\bold{a}_2(x)) - \gamma_m(\bold{a}_2(y))\right)\\
&\qquad + \gamma_z(\bold{a}_2(x))(\eta_1(x)-\eta_2(x)) - \gamma_z(\bold{a}_2(y))(\eta_1(y)-\eta_2(y))\\
&\qquad + (m_1-m_2)^2\left(R_{mm}(\bold{a}_1(x)\big|\bold{a}_2(x)) - R_{mm}(\bold{a}_1(y)\big|\bold{a}_2(y))\right)\\
&\qquad + R_{zz}(\bold{a}_1(x)\big|\bold{a}_2(x))(\eta_1(x)-\eta_2(x))^2 - R_{zz}(\bold{a}_1(y)\big| \bold{a}_2(y))(\eta_1(y)-\eta_2(y))^2\\
&\qquad + 2(m_1-m_2)\left(R_{mz}(\bold{a}_1(x)\big|\bold{a}_2(x))(\eta_1(x)-\eta_2(x)) - R_{mz}(\bold{a}_1(y)\big| \bold{a}_2(y))(\eta_1(y)-\eta_2(y))\right).
\end{split}
\end{equation}
Here, for $\bold{a}_1(x),\bold{a}_2(x)$, $x\in D_n$ we have set
\begin{equation*}
\begin{split}
R_{mm}(\bold{a}_1(x)\big|\bold{a_2}(x)) &= \int_0^1(1-u)\gamma_{mm}\left(\bold{a}_1(x) + u(\bold{a}_2(x)-\bold{a}_1(x))\right)du,\\
&=\int_0^1 (1-u)\gamma_{mm}\left(x, m_2 + u(m_1-m_2), \eta_2(x) + u(\eta_1(x)-\eta_2(x))\right)du,
\end{split}
\end{equation*}
with analogous formulas for $R_{zz}$ and $R_{mz}$. Since $m_2+u(m_1-m_2)$ is in between $m_1$ and $m_2$, and $\eta_2(x) + u(\eta_1(x)-\eta_2(x))$ is in between $\eta_1(x)$ and $\eta_2(x)$ this formula immediately gives (recall \eqref{E:bold_a_region})
\begin{equation}\label{E:h_taylor2}
\begin{split}
&\left|R_{mm}(\bold{a}_1(x)\big|\bold{a}_2(x))\right| \leq \frac{1}{2}\sup_{(x,m,z)\in E_n} |\gamma_{mm}(x,m,z)| = C(n,\alphanorm{\eta_1}{n},\alphanorm{\eta_2}{n}),
\end{split}
\end{equation}
(with analogous formulas for $R_{mz},R_{zz}$) as well as
\begin{equation}\label{E:h_taylor3}
\begin{split}
&\left| R_{mm}(\bold{a}_1(x)\big|\bold{a}_2(x)) - R_{mm}(\bold{a}_1(y)\big|\bold{a}_2(y))\right|\\
&\qquad \leq L_\gamma(M_n)\int_0^1(1-u)\left(|x-y| + |(1-u)(\eta_2(x)-\eta_2(y)) + u(\eta_1(x)-\eta_1(y))|\right)du,\\
&\qquad \leq \frac{1}{2}L_\gamma(M_n)\left(|x-y| + \alphanorm{\eta_2}{n}|x-y|^\alpha + \alphanorm{\eta_1}{n}|x-y|^\alpha\right),\\
&\qquad = C(n,\alphanorm{\eta_1}{n},\alphanorm{\eta_2}{n})|x-y|^\alpha,
\end{split}
\end{equation}
(with analogous formulas for $R_{zz},R_{mz}$ as well). We now use \eqref{E:h_taylor2}, \eqref{E:h_taylor3} to bound the five terms on the right hand side of \eqref{E:h_taylor1} separately. First,
\begin{equation*}
\begin{split}
&\left|(m_1-m_2)\left(\gamma_m(\bold{a}_2(x)) - \gamma_m(\bold{a}_2(y))\right)\right|\\
&\qquad \leq |m_1-m_2| L_\gamma(M_n)\left(|x-y| + \alphanorm{\eta_2}{n}|x-y|^\alpha\right),\\
&\qquad \leq C(n,\alphanorm{\eta_1}{n},\alphanorm{\eta_2}{n})|m_1-m_2||x-y|^\alpha.
\end{split}
\end{equation*}
Second
\begin{equation*}
\begin{split}
&\left|\gamma_z(\bold{a}_2(x))(\eta_1(x)-\eta_2(x)) - \gamma_z(\bold{a}_2(y))(\eta_1(y)-\eta_2(y))\right|\\
&\qquad \leq \left|\gamma_z(\bold{a}_2(x))\right|\left|\eta_1(x)-\eta_2(x) - (\eta_1(y)-\eta_2(y))\right| + \left|\eta_1(y)-\eta_2(y)\right|\left|\gamma_z(\bold{a}_2(x))-\gamma_z(\bold{a}_2(y))\right|,\\
&\qquad \leq \sup_{(x,m,z)\in \bar{E}_{M_n}}|\gamma_z(x,m,z)|\alphanorm{\eta_1-\eta_2}{n}|x-y|^\alpha\\
&\qquad\qquad + \alphanorm{\eta_1-\eta_2}{n}L_\gamma(M_n)\left(|x-y| + \alphanorm{\eta_2}{n}|x-y|^\alpha\right),\\
&\qquad = \calphanorm\alphanorm{\eta_1-\eta_2}|x-y|^{\alpha}.
\end{split}
\end{equation*}
Third, from \eqref{E:h_taylor3} we get
\begin{equation*}
\begin{split}
&(m_1-m_2)^2\left(R_{mm}(\bold{a}_1(x)\big| \bold{a}_2(x)) - R_{mm}(\bold{a}_1(y)\big|\bold{a}_2(y))\right)\\
&\qquad \leq C^{(1)}_nC(n,\alphanorm{\eta_1}{n},\alphanorm{\eta_2}{n})|m_1-m_2||x-y|^\alpha,\\
&\qquad = C(n,\alphanorm{\eta_1}{n},\alphanorm{\eta_2}{n})|m_1-m_2||x-y|^\alpha.
\end{split}
\end{equation*}
Fourth (recall \eqref{E:h_taylor2},\eqref{E:h_taylor3} and $a^2-b^2 = (a-b)(a+b)$)
\begin{equation*}
\begin{split}
&\left|R_{zz}(\bold{a}_1(x)\big|\bold{a}_2(x))(\eta_1(x)-\eta_2(x))^2 - R_{zz}(\bold{a}_1(y)\big|\bold{a}_2(y))(\eta_1(y)-\eta_2(y))^2\right|\\
&\qquad \leq \left|R_{zz}(\bold{a}_1(x)\big|\bold{a}_2(x))\right|\left|(\eta_1(x)-\eta_2(x))^2 - (\eta_1(y)-\eta_2(y))^2\right|\\
&\qquad\qquad + (\eta_1(y)-\eta_2(y))^2\left|R_{zz}(\bold{a}_1(x)\big|\bold{a}_2(x)) - R_{zz}(\bold{a}_1(y)\big|\bold{a}_2(y))\right|,\\
&\qquad \leq 2\calphanorm \alphanorm{\eta_1-\eta_2}{n}|x-y|^{\alpha}\\
&\qquad\qquad + \alphanorm{\eta_1-\eta_2}{n}^2 \calphanorm|x-y|^\alpha,\\
&\qquad = \calphanorm\alphanorm{\eta_1-\eta_2} |x-y|^{\alpha}.
\end{split}
\end{equation*}
Lastly, or fifth
\begin{equation*}
\begin{split}
&\left|2(m_1-m_2)\left(R_{mz}(\bold{a}_1(x)\big|\bold{a}_2(x))(\eta_1(x)-\eta_2(x)) - R_{mz}(\bold{a}_1(y)\big| \bold{a}_2(y))(\eta_1(y)-\eta_2(y))\right)\right|\\
&\qquad \leq 2|m_1-m_2||R_{mz}(\bold{a}_1(x)\big|\bold{a}_2(x))|\left|\eta_1(x)-\eta_2(x) - (\eta_1(y)-\eta_2(y))\right|\\
&\qquad \qquad + 2|m_2-m_2|\left|\eta_1(y)-\eta_2(y)\right|\left|R_{mz}(\bold{a}_1(x)\big|\bold{a}_2(x)) - R_{mz}(\bold{a}_1(y)\big|\bold{a}_2(y))\right|,\\
&\leq 2|m_1-m_2|\left(\calphanorm\alphanorm{\eta_1-\eta_2}{n}|x-y|^\alpha + \calphanorm|x-y|^\alpha\right)\\
&\qquad = \calphanorm |m_1-m_2|\alphanorm{\eta_1-\eta_2}|x-y|^{\alpha}.
\end{split}
\end{equation*}
Putting together the five estimates above in  \eqref{E:h_taylor1} we obtain
\begin{equation*}
\begin{split}
&\left|h^{m_1,\eta_1}(x) - h^{m_2,\eta_2}(x) - \left(h^{m_1,\eta_1}(y) - h^{m_1,\eta_1}(y)\right)\right|\\
&\qquad \leq \calphanorm \left(|m_1-m_2| + \alphanorm{\eta_1-\eta_2}{n} + |m_1-m_2|\alphanorm{\eta_1-\eta_2}{n}\right)|x-y|^{\alpha},
\end{split}
\end{equation*}
from which the result in \eqref{E:long_h_calc} follows.

\end{proof}

\section{Technical Results}

The following lemma shows that for all $m\geq 0$, the first time the balance $p(t,m)$ falls at or below $1/2$ is at least $T/2$:

\begin{lemma}\label{L:p_prop1}
For all $m>0$, $\inf{\bigl\{t \in [0,T]}:\  p(t,m) \leq (1/2)\bigr\} \geq T/2$.
\end{lemma}

\begin{proof}
Assume for some $m> 0,t\in [0,T]$, $p(t,m) = 2$. Then
\begin{equation*}
t=T+\frac{1}{m}\log\left(\frac{1}{2}\left(1+e^{-mT}\right)\right).
\end{equation*}
It is clear that
\begin{equation*}
t>\frac{T}{2} \iff \frac{1}{m}\log\left(\frac{1}{2}(1+e^{-mT})\right)>-\frac{T}{2} \iff \frac{1}{2}\left(1+e^{-mT}\right) > e^{-mT/2}.
\end{equation*}
The last inequality holds for all $m>0$ and $T>0$, finishing the proof.
\end{proof}


\section{On the Construction of the Risk Neutral Measure $\qprob$}\label{S:qprob}

Let $D$ be as in Assumption \ref{A:region} and let $\tilde{b}:D\mapsto\reals^d$ and $A:D\mapsto\mathbb{S}^d$ be given functions satisfying Assumption \ref{A:factor_coefficients}.  Assume that $D, \tilde{b}$ and $A$ are so that there exists a (necessarily unique) solution to the Martingale problem (see \cite{MR2190038}) for the second order linear operator $\tilde{L}$ associated to $(\tilde{b},A)$ on $D$. 

Now, fix a probability space $(\Omega,\mathcal{G},\prob)$ and denote by $\wt{W}$ a $d$-dimensional Brownian motion under $\prob$.  Set $\filtwt{W}$ as the $\prob$-augmented version of the right continuous enlargement of the natural filtration for $\widetilde{W}$, so that $\filt^{\widetilde{W}}$ satisfies the usual conditions. Since the Martingale problem for $\tilde{L}$ is well posed, there exists a unique strong solution to the SDE
\begin{equation}\label{E:factors_p}
dX_t = \tilde{b}(X_t)dt + a(X_t)d\widetilde{W}_t.
\end{equation}
where $a = \sqrt{A}$. Next let $\mu:D\mapsto\reals^d$, $\Sigma:D\mapsto\mathbb{S}^d$ also satisfy Assumption \ref{A:factor_coefficients}. With $\sigma = \sqrt{\Sigma}$, the market is formed via trading instruments $(S,S^0)$ where $S=(S^1,...,S^d)$ have dynamics
\begin{equation*}
\frac{dS^{i}_t}{S^i_t} = \mu^i(X_t) dt + \sum_{j=1}^{k} \sigma^{ij}(X_t) d\wt{W}^j_t;\qquad i = 1,...,d,
\end{equation*}
and $S^0_t = \xpn{\int_0^t r_u du}$ is the money market where $r = X^{(1)}$. Define $b:D\mapsto\reals^d$ by
\begin{equation}\label{E:mu_b_rep}
b(x) = \tilde{b}(x) - a(x)\sigma(x)^{-1}\left(\mu(x) - r\bold{1}\right),
\end{equation}
where $\bold{1}\in\reals^d$ is the vector of ones. Note that $b$ satisfies Assumption \ref{A:factor_coefficients}.  Lastly, assume the Martingale problem for $L$ associated to $(b,A)$ is also well posed on $D$. Under these hypotheses it is well known (see \cite[Ch. 5]{MR2057928}, \cite{MR1326606,MR2152242}) the above market (with $\filtwt{W}$ adapted, $S$-integral trading strategies) is complete, and the unique risk neutral measure $\qprob$ on $\mathcal{F}^{\widetilde{W}}_T$ has Radon-Nikodym derivative
\begin{equation}\label{E:Z_def}
\frac{d\qprob}{d\prob}\bigg|_{\mathcal{F}^{\widetilde{W}}_T} = Z_T;\qquad Z_t \dfn \mathcal{E}\left(-\int_0^{\cdot}\left(\mu(X_t)-r_t\bold{1}\right)'\sigma^{-1}(X_t)d\wt{W}_t\right)_t,\ t\leq T.
\end{equation}
In particular, $Z$ is a $(\prob,\filtwt{W})$ martingale. With $\qprob$ being well-defined on $\mathcal{F}^{\widetilde{W}}_T$, we recall (see \cite[Ch. 5]{MR2057928}) that, provided the requisite integrability holds, if $\mathcal{C} = \cbra{\mathcal{C}(t)}_{t\leq T}$ is a cumulative cash-flow stream, adapted to $\filt^{\widetilde{W}}$ and with rate $C(t) = \dot{\mathcal{C}}(t)$, then the unique price for the stream is given by $\espalt{\qprob}{\int_0^T C(t)e^{-\int_0^t r_udu}dt}$. With this notation in place, we now derive the mortgage price in two instances.

\subsection{Large Pool}\footnote{This derivation is alluded to, if not explicitly given, in \cite{MR2943181,MR2352905} and uses an argument similar to that in \cite{MR2116154}.} Assume that in addition to $\widetilde{W}$, $(\Omega,\G,\prob)$ supports an $\prob$-i.i.d. sequence of $U(0,1)$ random variables $\cbra{U_i}_{i=1,...}$ which are also $\prob$ independent of $\widetilde{W}$.  Let $\gamma$ be any non-negative, integrable, $\filt^{\widetilde{W}}$ adapted process. Given $\gamma$, the random times $\cbra{\tau_i}_{i=1,...}$ are constructed via
\begin{equation}\label{E:prepayment_times}
\tau_i = \inf\cbra{t\geq 0\ \such\ U_i = e^{-\int_0^t \gamma_u du}};\qquad i = 1,\dots.
\end{equation}
Note that the $\cbra{\tau_i}_{i\in I}$ are $\prob$ conditionally i.i.d. given $\F^W_T$, each with common $\prob$ - intensity $\gamma$.

Now, consider a large pool, consisting of infinitely many loans which are (uniformly) infinitely small.  More precisely, fix $N$ and for $i=1,...,N$ set $\tau_i$ as the prepayment time of the $i^{th}$ loan in an $N$-loan pool, with each loan of size $1/N$.  The pool has common contract rate $m$ and hence the respective principal balances and coupons are $p_i(t,m)= (1/N)p(t,m)$ (where $p$ is from \eqref{E:balance_closed_form_P01} and $c_i = (1/N)m/(1-e^{-mT}) = (1/N)c(m)$ for $i=1,...,N$. The cumulative cash flows of the pools is thus:
\begin{equation*}
\mathcal{C}_N(t) = \frac{1}{N}\sum_{i=1}^N c (t\wedge \tau_i) + \frac{1}{N}\sum_{i=1}^N p(\tau_i,m)1_{\tau_i \leq t}.
\end{equation*}
By the conditional law of large numbers and Glivenko-Cantelli type theorem in \cite[Theorem 6.6]{wellner2005empirical} we have that $\prob$-almost surely:
\begin{equation*}
\lim_{N\uparrow\infty} \sup_{t\in [0,T]}\left|\mathcal{C}_N(t)-\mathcal{C}(t)\right| = 0,
\end{equation*}
where for $t\leq T$ and $\tau$ a generic copy of $\tau_i$:
\begin{equation*}
\begin{split}
\mathcal{C}(t) &= c\condespalt{}{t\wedge\tau}{\F^{\widetilde{W}}_T} + \condespalt{}{p(\tau,m)1_{\tau\leq t}}{\F^{\widetilde{W}}_T},\\
&=ct e^{-\int_0^t\gamma_udu} + c\int_0^t u\gamma_u e^{-\int_0^u\gamma_v dt}du + \int_0^t p(u,m)\gamma_u e^{-\int_0^u\gamma_vdv}du.
\end{split}
\end{equation*}
The cash flow rate is
\begin{equation*}
C(t) = ce^{-\int_0^t\gamma_u du} + p(t,m)\gamma_t e^{-\int_0^t\gamma_u du}.
\end{equation*}
It thus follows that the price of the large pool is given by
\begin{equation*}
\begin{split}
\espalt{\qprob}{\int_0^T(c+p(t,m)\gamma_t)e^{-\int_0^t(r_u+\gamma_u)du}dt} = 1+ \espalt{\qprob}{\int_0^T(m-r_t)p(t,m)e^{-\int_0^t (r_u+\gamma_u)du}dt},
\end{split}
\end{equation*}
where the last inequality follows by using \eqref{E:balance_no_prepay} and integration by parts.  This yields \eqref{E:mortgage_value_nice} and $\gamma$ is the $\prob$ prepayment intensity.

\subsection{Single Loan Pool}

Here, we assume that in addition to $\wt{W}$, $(\Omega,\G,\prob)$ supports a $U(0,1)$ random variable $U$ which is $\prob$ - independent of $\wt{W}$. The random time $\tau$ is created as in \eqref{E:prepayment_times} where $\gamma$ is again a non-negative, integrable, $\filt^{\wt{W}}$ adapted process. Associated to $\tau$ is the indicator process $H = \cbra{H_t}_{t\geq 0}$ with $H_t = 1_{\tau> t}$.  $H$ generates the filtration $\filt^{H} = \cbra{\mathcal{H}_t}_{t\geq 0}$ via $\mathcal{H}_t = \sigma(H_s;s\leq t)$ and  $\tau$ is clearly an $\filt^{H}$-stopping time.  Furthermore, $\filt^{H}$ and $\filt^{\wt{W}}$ are $\prob$ independent. Lastly, the enlarged filtration $\mathbb{G}$ is that generated by both $\filt^{\wt{W}}$ and the $\prob$-augmented versions of $\filt^{H}$, and is right continuous \cite[Theorem 1]{MR671249}.  Now, let $A\in \F^{\wt{W}}$ and $t\geq 0$. We clearly have that $\espalt{\prob}{1_{\tau > t}1_A} = \espalt{\prob}{(1-e^{-\int_0^t\gamma_udu})1_A}$ and hence
\begin{equation*}
\condprobalt{\prob}{\tau>t}{\F^{\wt{W}}} = \condprobalt{\prob}{\tau > t}{\F^{\wt{W}}_t} = 1-e^{-\int_0^t \gamma_u du},
\end{equation*}
so that $\gamma$ is the $(\prob,\filt^{\wt{W}})$ intensity of $\tau$. Enlarge the market described above to allow for $\mathbb{G}$ adapted trading strategies.  Though this market is now incomplete, it follows that the minimal entropy martingale measure $\qprob$ (same notation as above) satisfies
\begin{equation*}
\frac{d\qprob}{d\prob}\bigg|_{\mathcal{G}_T} = Z_T;\qquad T\geq 0.
\end{equation*}
Indeed, this fact has been shown in \cite{MR2011941, Robertson_Spil_2014} amongst others.  We next claim that $\gamma$ is the $\qprob$ intensity of $\tau$ as well. To see this note that $U\sim U(0,1)$ under $\qprob$ since $\qprob\bra{U\leq u} = \espalt{\prob}{1_{U\leq u}Z_T} = \prob\bra{U\leq u} = u$. Next, $U$ is $\qprob$ independent of $\filt^{\wt{W}}$ since for all $A\in \F^{\wt{W}}_T$ for any $T\geq 0$:
\begin{equation*}
\qprob\bra{U\leq u, A} = \espalt{\prob}{1_{U\leq u}1_A Z_T} = \prob\bra{U\leq u}\qprob\bra{A} = \qprob\bra{U\leq u}\qprob\bra{A},
\end{equation*}
and hence the $\qprob$ independence follows. Thus, for all $A\in \filt^{\wt{W}}$ and $t\geq 0$:
\begin{equation*}
\qprob\bra{\tau>t, A} = \espalt{\qprob}{1_A \condespalt{\qprob}{1_{U>e^{-\int_0^t\gamma_udu}}}{\filt^{\wt{W}}}} = \espalt{\qprob}{1_A\left(1-e^{-\int_0^t \gamma_u du}\right)},
\end{equation*}
proving that $\gamma$ is the $\qprob$ intensity of $\tau$. Now, starting with the price for the mortgage as in \eqref{E:mortgage_value} where $\qprob$ is now the minimal entropy measure in the enlarged market, equation \eqref{E:mortgage_value_nice} still holds (see \eqref{E:intensity_equation}) and hence \eqref{E:mortgage_value_nice} and \eqref{E:m_goal} hold.

\bibliographystyle{siam}
\bibliography{/home/scottrob/Bibliography/master.bib}

\end{document}